%% file: main.tex
\documentclass[a4paper, 10pt]{article}
\usepackage[margin=1in, dvips]{geometry}

\usepackage{mathtools,amssymb,amsfonts,mathrsfs,amsthm,thmtools}
\usepackage{graphicx,tikz, tikz-cd}
\usepackage{pdfpages}
\usepackage[final]{microtype}
\usepackage[pdftex]{hyperref}

\usepackage{cleveref}

\usepackage[final]{showkeys}
\usepackage[
    backend=biber,
    style=alphabetic,
    url=true
  ]{biblatex}
\usepackage{xurl}
\hypersetup{breaklinks=true}

\usepackage{datetime}

\usepackage{array}
\usepackage[all]{nowidow}
\usepackage[affil-it]{authblk}

\usepackage{eqparbox}
\usepackage{xparse}

\usepackage{enumitem}

\usepackage[def-file=diffcoeff5]{diffcoeff}
\usepackage{tensor}

\usepackage{subfiles}

\usepackage[colorinlistoftodos,prependcaption,textsize=tiny, disable]{todonotes}

\numberwithin{equation}{section}

\allowdisplaybreaks

\theoremstyle{plain}
\newtheorem{theorem}{Theorem}[section]

\newtheorem{lemma}[theorem]{Lemma}
\newtheorem{proposition}[theorem]{Proposition}

\theoremstyle{definition}

\newtheorem{definition}[theorem]{Definition}
\newtheorem*{acknowledgements}{Acknowledgements}

\theoremstyle{remark}
\newtheorem{remark}[theorem]{Remark}

\newcommand{\R}{\mathbb{R}}
\newcommand{\N}{\mathbb{N}}
\newcommand{\sph}{\mathbb{S}}

\newcommand{\mc}[1]{\mathcal{#1}}
\renewcommand{\rm}[1]{\mathrm{#1}}
\newcommand{\mf}[1]{\mathfrak{#1}}

\newcommand{\nT}{\tensor[^0]{T}{}}

\newcommand{\sfG}{\mathsf{G}}
\newcommand{\sR}{\mathscr{R}}

\newcommand{\dd}{\rm{d}}

\DeclareMathOperator{\Diff}{Diff}
\DeclareMathOperator{\Ric}{Ric}

\DeclareMathOperator{\Tr}{tr}
\DeclareMathOperator{\Ker}{ker}

\newcommand{\del}{\partial}
\newcommand{\deld}{\delta \rm{d}}

\DeclarePairedDelimiter\abs{\lvert}{\rvert}

\renewcommand{\epsilon}{\ensuremath\varepsilon}

\renewcommand{\phi}{\ensuremath{\varphi}}

\begin{document}
\title{Gluing charged black holes into de Sitter space}
\author[1,2]{Markus Verlemann \thanks{mv582@cam.ac.uk}}
\affil[1]{\small University of Cambridge, Center for Mathematical Sciences, Wilberforce Rd, Cambridge CB3 0WA, UK}
\affil[2]{\small ETH Zürich, 
Rämistrasse 10, 8092 Zurich, Switzerland}
    \maketitle
    \input{abstract}
   \input{Introduction}
   \section{Preliminaries}
   \input{0-geometry}

   \input{differentialoperators}

   \input{EinsteinMaxwell}
   \input{manifolds}
   \input{LinearizationOfEinsteinMaxwell}
   \section{The Gluing Construction}
   \input{gluingexplanation}

\input{LeadingOrderCorrectionObstruction}
   \input{toinfinityandbeyond}
   \input{solvingthenonlinearequations}

   \input{finalremarks}
   \input{KerrNewmanDeSitter}
   \printbibliography
\end{document}

%% file: abstract.tex
\begin{abstract}
    We extend Hintz's cosmological black hole gluing result to the Einstein--Maxwell system with positive cosmological constant by gluing multiple Reissner--Nordström or Kerr--Newman--de Sitter black holes into neighbourhoods of points in the conformal boundary of de Sitter space. We determine necessary and sufficient conditions on the black hole parameters -- related to Friedrich's conformal constraint equations -- for this gluing to be possible. We also improve the original gluing method slightly by showing that the construction of a solution in Taylor series may be accomplished using an exactness argument, eliminating the need for an early gauge-fixing.
\end{abstract}

%% file: Introduction.tex
\section{Introduction}
In the theory of relativity, gravity is modeled as the curvature of a spacetime $M$ with Lorentzian metric $g$ of signature (${-}{+}{+}{+}$). The \emph{Einstein field equations}, given by
\begin{equation}\label{eq:IntroEinstein}
\Ric(g) - \frac{R}{2}g + \Lambda g = 2T,
\end{equation}
relate the curvature components of the metric $g$ to the flux of energy and matter in the spacetime, described by the energy-momentum tensor $T$. This curvature in turn affects the shape of `straight' lines in the spacetime and as such particles' motions. \par
Another fundamental force, electromagnetism, is characterized by the behavior of the electromagnetic tensor $F$. This $2$-form is required to satisfy \emph{Maxwell's equations}, which in the case where there are no sources or currents present on $M$ take the form
\begin{equation}\label{eq:IntroMaxwell}
\dl{F} = 0, \quad \dl \star F= 0.
\end{equation}
Here $\star$ is the Hodge star operator associated to $g$. Such an electromagnetic field comes with an energy-momentum tensor of the form
\[
T_{\mu\nu} = F\indices{_\mu^\alpha}F\indices{_\nu_\alpha } - \frac14 F^{\alpha\beta}F_{\alpha\beta}g.
\]
Taken together, these equations therefore form a system of differential equations coupling $g$ and $F$ to each other. This system is called the \emph{Einstein-Maxwell system}. \par
A class of solutions which is especially of interest are the so-called black hole solutions. In \cite{blackholegluing} Hintz introduced a flexible gluing method by constructing spacetimes with multiple \emph{uncharged} black holes and cosmological constant $\Lambda>0$. Here we extend these results by gluing charged black holes. In particular, our spacetimes are closely related to the Kastor--Traschen spacetimes \cite{PhysRevD.47.5370}, which describe multiple black holes with charge equal to their mass (which are subextremal for $\Lambda > 0$). \par
The simplest solution to the Einstein-Maxwell equations with $\Lambda > 0$, $F=0$ is de Sitter space, which describes an expanding universe. In four dimensions it may be viewed as the manifold
\[
M = \Big(-\frac{\pi}{2}, \frac{\pi}{2}\Big)_s \times \sph^3, \quad g_{dS} = \frac{3}{\Lambda} \frac{-\dl{s}^2 + g_{\sph^3}}{\cos^2(s)}.
\]
To investigate the behavior of such spacetimes in the far future we may look at the partial conformal compactification
\[
M = \Big(-\frac{\pi}{2}, \frac{\pi}{2}\Big]_s \times \sph^3, \quad g = \cos^2(s) g_{dS} 
\]
and consider the future timelike infinity $\{\pi/2\}_s \times \sph^3$ with metric induced by $g$. Such conformal changes do not affect the causal structure of the spacetime. \par
A non-rotating charged black hole is described by the Reissner--Nordström--de Sitter metric. As demonstrated later, such a black hole may be regarded as a black hole in a de Sitter universe, drifting toward some point $p$ at the future conformal boundary of de Sitter space. 
For charge $Q \in \R$ and mass $\mf{m} \in \R$ it is modelled by the metric
\begin{equation}\label{eq:RNdSIntro}
        g_{\mf{m},Q} = -\mu_{\mf{m},Q}(r) \dl{t:2} + \mu_{\mf{m},Q}(r)^{-1}\dl{r:2} + r^2 g_{\sph^2}  
    \end{equation}
    where $\mu_{\mf{m},Q}(r) = 1-\frac{2m}{r} + \frac{Q^2}{r^2} - \frac{\Lambda}{3}r^2$. In addition to the metric, the black hole comes equipped with an electromagnetic potential, $A_Q = -Q/r \dl{t}$, whose exterior derivative gives us the electromagnetic tensor $F= \dl{A_Q}$. Together, the metric and potential make up a solution to the Einstein--Maxwell equations. \par
Taking the charge $Q$ to be $0$, so $F=0$, one obtains the Schwarzschild--de Sitter spacetime. In \cite{blackholegluing}, Hintz found that one may `glue' various Schwarzschild--de Sitter black holes into the far future of de Sitter space, as long as the black hole masses satisfy a \emph{balance condition}. We provide a direct generalization for Reissner--Nordström--de Sitter black holes. Specifically, we derive sufficient -- and under some conditions, necessary -- balance conditions for the masses and chargees under which one is able to `glue' Reissner-Nordström black holes into de Sitter space. More precisely, we prove:
\begin{theorem}
    Let $N \in \N$ and let $(p_i, \mf{m}_i, Q_i) \in \sph^3 \times \R \times \R, 1\leq i \leq N$, satisfy the \emph{charge balance condition}
    \begin{equation}
        \sum_{i=1}^N Q_i = 0,
    \end{equation}
    and the \emph{mass balance condition}
    \begin{equation}
        \sum_{i=1}^N \mf{m}_i p_i = 0 \in \R^4.
    \end{equation}
    Then there exists a metric $g$ and an electromagnetic potential $A$ with the following properties
    \begin{enumerate}
        \item $g$ and $F= \dl{A}$ satisfy the Einstein--Maxwell equations \eqref{eq:IntroEinstein}, \eqref{eq:IntroMaxwell};
        \item in a neighborhood of each $p_i$, the metric $g$ and potential $A$ are equal to the Reissner-Nordström-de Sitter metric and potential with mass $\mf{m}_i$ and charge $Q_i$;
        \item outside a neighborhood of $\{p_1,\dots, p_N\}$, the conformally rescaled metric $\cos^2(s)g$ asymptotes to $\cos^2(s)g_{dS}$ at the rate of $\cos^3(s)$ and $\cos(s)A$ approaches the de Sitter potential $0$ at the rate $\cos^2(s)$.
    \end{enumerate}
\end{theorem}
As in Hintz's gluing procedure, we glue the cosmological region of the Reissner--Nordström--de Sitter black holes into de Sitter space. This is the region $r>r_+$, where $r_+$ is the largest root of $\mu(r)$ (when $\mu$ has 4 distinct roots). As any roots of $\mu(r)$ are only coordinate singularities, gluing only this region suffices, as one may `complete' the black holes later; compare with \cite[Figure 1.3.]{blackholegluing}.\par
We also show how some additional observations on Schwarzschild--de Sitter gluing made in \cite{blackholegluing} generalize to Reissner--Nordström--de Sitter gluing in a straightforward manner:
\begin{itemize}
    \item For very small masses and charges the cosmological horizons of at least two Reissner-Nordström black holes must intersect, see \cref{remark:DomainOfExistence};
    \item Under some vanishing conditions on the corrections to the de Sitter metric and potential, the mass and charge balance conditions are necessary, see \cref{theorem:NecessityOfBalanceConditions};
    \item If one allows $g,A$ to become singular at a point $p_\infty \in \sph^3$, the balance conditions are \emph{not} needed, see \cref{theorem:BlackHoleGluingNonCompact}.
\end{itemize}
We last show in \cref{sec:RotatingChargedBlackHoles} that the gluing procedure works for rotating charged black holes, modeled by the Kerr--Newman--de Sitter metric and electromagnetic potential. These have an additional rotation parameter to keep track of, which changes the balance condition slightly. The charge and mass balance conditions will instead involve a scaled version of the charges and masses, whose factor depends on the speed of rotation. Additionally, the rotation parameters will have to sum to zero in a way specified in \cref{theorem:blackholegluingKNdS}.

\subsection{Gluing at the conformal boundary and the constraints}
In the initial value formulation of general relativity, spacetimes are constructed as developments of initial data. For the Einstein--Maxwell equations, such initial data consist of a tuple $(\Sigma, h, k, \mathbf{E}, \mathbf{B})$, where $\Sigma$ is a $3$-manifold with a Riemannian metric $h$ and $k$ is a symmetric $2$-tensor. The initial data has to satisfy the constraint equations
\begin{gather*}
    R_h + (\Tr_h k)^2 - \abs{k}_h^2 - 2\Lambda = 2(\abs{\mathbf{E}}_h^2 + \abs{\mathbf{B}}_h^2), \\
    \delta_h k + \dl{\Tr_h} k = 2 \star_h (\mathbf{E} \wedge \mathbf{B}), \\
    \delta_h \mathbf{E} = 0, \quad \delta_h \mathbf{B} = 0,
\end{gather*}
where $R_h$ is the scalar curvature and $\delta_h k_\nu = - \nabla^\mu k\indices{_\mu_\nu} $ is the divergence operator of $h$. When $\Sigma \subset M$ is a spacelike hypersurface and $g,F$ are given on $M$, then $h$ is the metric on $\Sigma$ induced by $g$, and $k$ the second fundamental form.  The covectors $\mathbf{E}$ and $\mathbf{B}$ are then electric and magnetic field induced by the electromagnetic tensor $F$ on $\Sigma$. Specifically,
\[
\mathbf{E} = \star_h i_{\Sigma}^* \star_g F, \quad \mathbf{B} = \star_h i_\Sigma^* F,
\]
where $\star_h, \star_g$ are the Hodge stars of the respective metrics and $i_\Sigma^*$ denotes the pullback to $\Sigma$.
The first two constraints are found by evaluating the normal-normal and normal-tangential components of the Einstein field equations on $\Sigma$, the latter two by pulling back Maxwell's equations to $\Sigma$. If $F=\dl{A}$ for some electromagnetic potential $A$, then the constraint $\delta_h \mathbf{B} =0$ is automatically satisfied. \par
As was discovered by Choque-Bruhat in \cite{ChoqueBruhat1952}, for specified initial data satisfying the constraint equations, there exists a unique local solution to the Einstein-Maxwell system. In other words, the Cauchy problem in general relativity is well posed. Moreover, every initial data set admits a unique maximal globally hyperbolic development \cite{Choquet-Bruhat1969Dec}. A more in-depth discussion of this is provided for example \cite[\S 6]{GR-bruhat}. \par
Most gluing constructions prepare initial data so that they model the desired physical phenomena, such a black hole with some asymptotically flat behavior at infinity as in \cite{Cor00, CS06}, wormholes as in \cite{IMP02,IMP03}, or multiple Schwarzschild or Kerr black holes as in \cite{CD02, CD03}. We leave a more detailed discussion of these and how they pertain to Hintz's gluing construction to \cite[\S 1.1.]{blackholegluing}. Also see the work by Carlotto--Schoen\cite{CarlottoSchoen16}, which uses a gluing method to construct initial data localized to cones, and \cite{mao2023localizedinitialdataeinstein, mao-oh-tao} in which gluing methods using explicit solution operators are presented. Other recent results include \cite{hintz2022gluingsmallblackholes} and the series of papers \cite{hintz2024gluingsmallblackholesI, hintz2024gluingsmallblackholesII,hintz2024gluingsmallblackholesIII}, where geometric microlocal techniques are used to glue rescaled black holes into smooth initial data and along timelike geodesics respectively.\par
For studying the Einstein equations with $\Lambda > 0$, conformal methods have proven particularly useful. For example, one may consider the Einstein vacuum equations if one conformally changes the metric. As first noticed by Friedrich in \cite{Friedrich1986}, the resulting conformal Einstein equations and conformal constraints simplify a lot at the conformal boundary in the case $n+1=4$.
To elaborate on this, recall that an asymptotically simple spacetime \cite{Roger1965Feb} with cosmological constant $\Lambda > 0$ is a spacetime with boundary defining function $\tau$ and metric $g$ such that $\tau^2g$ is smooth up to the conformal boundary $\tau = 0$ and such that the metric $h$ induced by $\tau^2g$ at $\tau = 0$ is Riemannian. Both the de Sitter spacetime and the cosmological region of the Reissner--Nordström--de Sitter are asymptotically simple. In the Einstein vacuum case, scattering data at the conformal boundary then consist of the conformal boundary $\Sigma$ with Riemannian metric $h$, together with a traceless symmetric $2$-tensor $k$ satisfying
\begin{equation}\label{eq:introconstraintk}
\delta_h k = 0.
\end{equation}
This is the \emph{only} differential constraint that has to be satisfied at the conformal boundary. \par
Maxwell's equations on the other hand are already conformally invariant. On the level of an electromagnetic potential, the only constraint from Maxwell's equations therefore is
\begin{equation}\label{eq:introconstraintE}
\delta_h \mathbf{E} = 0,
\end{equation}
where $\mathbf{E}$ is the electric field induced at the conformal boundary as before. Now, it will be the case for Reissner--Nordström black holes that the energy-momentum tensor vanishes at the rate $\cos^2(s)$. This will imply that the \emph{rescaled Cotton tensor} vanishes at the conformal boundary $s = \pi/2$, cf. \cite[\S 9.1.2]{Kroon2016Jul}. Important for our purposes is then that this vanishing implies that the constraint \eqref{eq:introconstraintk} remains unchanged when considering the full Einstein-Maxwell system, see [\S 11]\cite{Kroon2016Jul}. \par
In summary, the constraints for the Einstein--Maxwell system with fast vanishing energy-momentum tensor reduce to the two \emph{linear} equations \eqref{eq:introconstraintk} and \eqref{eq:introconstraintE}. This indicates that one may try the following procedure for gluing an asymptotically simple spacetime into de Sitter space. \par
Say we have asymptotically simple metrics $g_i$ and potentials $A_i$ which we aim to glue into neighborhoods of some points $p_i$ at the conformal boundary of de Sitter space. We may start by considering a preliminary glued metric and potential obtained via a simple gluing using a partition of unity. There is no reason this should solve the Einstein--Maxwell equations, and the constraints will only be solved up to some error terms $f_1, f_2$. We may however hope to correct these errors by adding some terms to the glued metric and potential which \emph{do not} change the metric and potential in a region around the $p_i$. This way, we do not affect the black holes except in their cosmological regions.
\par
To accomplish this, one needs to solve the underdetermined elliptic differential equations
\begin{align}
    \delta_h k' = f_1 \label{eq:introdiffeqE} \\
    \delta_h \mathbf{E}' = f_2 \label{eq:introdiffeqk}
\end{align}
for $k'$ symmetric traceless and $\mathbf{E}'$ a covector. Necessary for solvability is that the errors $f_1$ and $f_2$ are orthogonal to the cokernels of $\delta_h$, which are the conformal Killing vector fields in the case of $\delta_h$ acting on traceless symmetric $2$-tensors and the locally constant functions in the case of $\delta_h$ acting on covectors. A result by Delay \cite{Delay2012} shows that these conditions are actually sufficient for solvability with the desired restrictions on the supports -- in other words, that the glued metric and potentials remain unchanged in a region around the $p_i$. In the $\mathbf{E}'$ case, one may alternatively use cohomological methods, which is also the route we will take. As in \cite{blackholegluing}, the solvability condition for the $k'$ equation will lead to the balance condition for the mass, and in our case the condition on $f_2$ will lead to the additional charge balance condition.
\subsection{Gluing in a 0-geometry framework}
In the approach taken by \cite{blackholegluing} one does not work with Friedrich's conformal Einstein equations directly. Rather, one views asymptotically simple metrics as cases of Lorentzian $0$-metrics in Mazzeo--Melrose's \cite{MazzeoMelrose} $0$-framework, which allows for a direct treatment of objects degenerating at the conformal boundary. \par
We will use the same methods and give a brief introduction to this in \cref{section:0-geometry}. As discussed in the previous subsection, we first naively glue Reissner--Nordström--de Sitter black holes into de Sitter space using a partition of unity subordinate to a cover of a neighborhood of the conformal boundary, $\sph^3$, by neighborhoods of the $p_i$. Call the thus obtained metric $g_{(3)}$ and the potential $A_{(2)}$, named this way for the errors to the Einstein--Maxwell system they produce. Indeed, the error to the inhomogeneous Maxwell's equations is a $\mc{O}(\tau^3)$ multiple of the $0$-covector field $d\tau/\tau$ and supported away from the $p_i$. A variant of the Maxwell constraint becomes relevant for correcting this error. Namely, an investigation of the linearization of Maxwell's equations shows that adding an order $\tau^3$ term to $g_{(3)}$ or $A_{(2)}$ does not affect this error. Instead, we correct this error on the $\tau^2$ level by adding such a term to $A_{(2)}$. An investigation of leading order terms of the linearization shows that this does not produce an order $\tau^2$ error; correcting the error then just means solving a differential equation as in \eqref{eq:introdiffeqE}.
We then call the corrected potential $A_{(3)}$, for it solves Maxwell's equations to order $\tau^4$. \par
The error to Einstein's field equations will be seen to be of size $\mc{O}(\tau^4)$ in the $0$-geometry framework. Again, adding an order $\tau^4$ term to the metric or potential does not help in correcting this, and we use a $\tau^3$ correction to the metric instead. To ensure that this does not disturb Einstein's equation on the $\tau^3$ level, the correcting $2$-tensor will need to be traceless and solve an equation as in \eqref{eq:introdiffeqk}. We call the new metric $g_{(4)}$. \par
After this, solving the Einstein-Maxwell system in Taylor series poses fewer problems; in particular no further differential constraints arise. We continue the scheme of alternately finding a correction to the potential which increases the order of vanishing of Maxwell's equations, and a correction to the metric which increases the order of Einstein's field equations. We call the resulting metric $g_{(\infty)}$ and potential $A_{(\infty)}$ -- these solve the Einstein--Maxwell system to infinite order at the conformal boundary. \par
We then finish the construction by finding a true solution $(g,A)$ to the Einstein--Maxwell system by gauge-fixing. We use a DeTurck gauge with background metric $g_{(\infty)}$ (see \cite{DeTurck81} and \cite{GrahamLee}) for the Einstein part and a Lorenz-like gauge with background potential $A_{(\infty)}$ (see \cite[\S 2.2]{HintzKerrNewmanDeSitter}) to transform the system into a set of quasilinear wave equations in a de Sitter background. Results from \cite{asy-ds} then provide the local solvability of this system near the boundary $\sph^3$ by a rapidly vanishing correction to the background values $g_{(\infty)}, A_{(\infty)}$. 

\begin{acknowledgements}
This paper was written as a Master's thesis at ETH Zürich. I express my sincerest gratitude to my advisor, Peter Hintz, for his invaluable ideas, expertise, and feedback. His support in finding a PhD to continue my academic work in general relativity has also helped me tremendously. I also thank my friends for the enthusiastic discussions while working on our theses together. The last revisions of this paper were done with support from EPSRC scholarship EP/W524633/1.
\end{acknowledgements}

%% file: 0-geometry.tex
\subsection{0-geometry}\label{section:0-geometry}
    We start by introducing some aspects of Mazzeo--Melrose's \cite{MazzeoMelrose} $0$-framework.
    Throughout this section, let $M$ be an $(n+1)$-dimensional manifold with boundary $\del M \neq 0$. A boundary defining function is a function $\tau \in \mc{C}^\infty(M)$ with $\tau^{-1}(0) = \del M$ and $\dl{\tau} \neq 0$ on $\del M$.
    \begin{definition}
         We define the space of $0$-vector fields to be
        \[
            \mc{V}_0(M) = \{V \in \mc{V}(M) : V(p) = 0 \ \forall \ p \in \del M\},
        \]
        where $\mc{V}(M) = \mc{C}^\infty(M; TM)$ denotes the space of smooth vector fields on $M$. We define the $0$-tangent bundle $\nT M$ to be the smooth vector bundle over $M$ whose frames are given in a boundary chart $(\tau, x^i), 1 \leq i \leq n$ by \[
        \tau\del_\tau,\tau\del_{x^i}, 1\leq i \leq n.
        \]
        $0$-vector fields are then just the sections of this bundle, so $\mc{V}_0(M) = \mc{C}^\infty(M; \nT M)$. \\
        $0$-covector fields are defined in the same fashion as sections of the vector bundle $\nT^*M$, which has frames given by
        \[
        \frac{\dl{}\tau}{\tau}, \frac{\dl{x:i}}{\tau}, 1\leq i \leq n
        \]
        in local coordinates $[0,\infty)_\tau \times \R^n_x$. \\
        We also define a smooth Lorentzian $0$-metric to be a smooth section $\mc{C}^\infty(M; S^2 \nT^*M)$ with signature $(n,1)$ everywhere. One may similarly define $0$-metrics with other degrees of smoothness if needed.
    \end{definition}
    We use the opportunity to give a small lemma giving some intuition about boundary defining functions.
    \begin{lemma}\label{lemma:boundarydefiningfunctions}
        Let $\tau$ and $\Tilde{\tau}$ be two boundary defining functions on $M$. Let $x$ be a chart of the boundary around a point $p \in \del M$. Then in a small neighborhood of $p$ in $M$ we may take $(\tau,x)$ and $(\Tilde{\tau},x)$ as boundary charts near $p$. We also have $\tau\del_\tau \vert_{\del M} = \Tilde{\tau}\del_{\Tilde{\tau}} \vert_{\del M}$ as $0$-vector fields.
    \end{lemma}
    \begin{proof}
        From the formula for a change of coordinates we have
        \[
        \tau\del_\tau = \frac{\tau}{\Tilde{\tau}}\diffp{\Tilde{\tau}}{\tau}\Tilde{\tau}\del_{\Tilde{\tau}} + \diffp{x_i}{\tau}\tau\del_{x_i}
        \]
        As $\diffp{x_i}{\tau} = 0$ the claim now follows from using L'Hôpital's rule on the quotient $\frac{\tau}{\Tilde{\tau}}$ which shows that $\frac{\tau}{\Tilde{\tau}}$ is a positive smooth function.
    \end{proof}

     For $m \in \N_0$, we define $\Diff_0^m(M)$, the space of $0$-differential operators of order $m$, to be the vector space of linear combinations up to $m$-fold compositions of 0-vector fields. These may be viewed as acting on functions or, more generally, on sections of vector bundles. We recall some $0$-differential operators in \cref{sec:Differentialoperators}. \par
     We next discuss indicial families, which capture leading order behaviour of differential operators.
    \begin{definition}[Indicial families]
        Let $\tau$ be a boundary defining function. Let \[A = \sum_{i + |\alpha| \leq m} A_{i\alpha}(\tau,x) (\tau\del_\tau)^i(\tau\del_x)^\alpha\in \Diff_0^m(M),\] where the $\alpha$ are multi-indices. By commuting the $\tau$ terms coming from the $\tau\del_x$ terms with the $\tau\del_\tau$ terms we may alternatively write this in the form
        \[
        A = \sum_{i + |\alpha| \leq m} a_{i\alpha}(\tau,x)\tau^{\abs{\alpha}}(\tau\del_\tau)^i\del_x^\alpha.
        \]
        If we additionally expand $a_{i\alpha}(\tau,x) = \sum_{l=0}^n  a_{i\alpha}^{(l)}(x) \tau^l + \mc{O}(\tau^{n+1})$, we can finally write
        \[ 
            A \equiv \sum_{k=0}^n \tau^k\sum_{\genfrac{}{}{0pt}{}{i+|\alpha| \leq m}{|\alpha| \leq k}} a_{i\alpha}^{(k-|\alpha|)}(x) (\tau\del_\tau)^i\del_x^\alpha.
        \]
        modulo $\tau^{n+1} a(\tau,x)(\tau\del_\tau)^i \del_x^{\alpha}$ terms. We then define, for $k \in \N_0$ and $\lambda \in \R$
        \begin{equation}
            I(A[\tau^k],\lambda) \coloneqq \sum_{\genfrac{}{}{0pt}{}{i+|\alpha| \leq m}{|\alpha| \leq k}} a_{i\alpha}^{(k-|\alpha|)}(x) \lambda^i\del_x^\alpha.
        \end{equation}
        Observe that $I(A[\tau^0],\lambda)$ does not depend on the choice of boundary defining function; we henceforth write $I(A,\lambda) \coloneqq I(A[\tau^0],\lambda)$. 
    \end{definition}
    The indicial families are defined in such a way that if $u = u(x) \in \mc{C}^\infty(\del M)$, then
    \[
    A(\tau^\lambda u) = \tau^{\lambda}I(A,\lambda)u + \tau^{\lambda +1}I(A[\tau],\lambda)u + \dots + \tau^{\lambda + n}I(A[\tau^n],\lambda)u + \mc{O}(\tau^{\lambda + n + 1}).
    \]

%% file: differentialoperators.tex
\subsection{Differential operators}\label{sec:Differentialoperators}
Let $M$ be an $n$-dimensional oriented manifold with smooth Lorentzian or Riemannian $0$-metric $g$. Many of the geometric operators familiar from differential geometry are $0$-differential operators. For example, taking a Levi--Civita derivative of a $0$-vector field with respect to another $0$-vector field again produces a $0$-vector field. One may directly calculate this using Christoffel symbols, or in a coordinate invariant way by observing that Lie brackets of $0$-vector fields are $0$-vector fields and using Koszul's formula. \par
We give some examples of such differential operators here, all of which will be encountered later. First recall the Hodge star operator $\star$, acting on a $k$-form $\beta$. A choice of orientation of $M$ provides a nowhere vanishing $n$-form $\rm{dvol}_g$. The Hodge star is then defined implicitly by requiring
\[
\alpha \wedge (\star \beta) = g(\alpha,\beta) \rm{dvol}_g
\]
for all $k$ forms $\alpha$, where $g(\alpha,\beta)$ denotes the canonical inner product on $k$-forms defined by $g$. The Hodge star is an isomorphism from the space of $k$-forms $\mc{C}^\infty (M, \Lambda^k T^*M)$ to the space of $n-k$ forms $\mc{C}^\infty (M, \Lambda^{n-k} T^*M)$. The codifferential $\delta_g$ is then the differential operator acting on a $k$-form $\beta$ using the formula
\[
    \delta_g \beta = (-1)^{n(k+1)+1}s\star \rm{d} \star \beta,
\]
where $\rm{d}$ is the exterior derivative and $s$ is the signature of the metric $g$ (so $s=1$ if $g$ is Riemannian and $s=-1$ if it is Lorentzian). This agrees with the alternative definition as the divergence, defined more generally on smooth $k$-tensors,
\[
    \delta_g \beta_{\mu_2, \dots, \mu_k} = - \nabla^{\mu_1}\beta_{\mu_1 \dots \mu_k}.
\]
The formal adjoint of the divergence acting on symmetric $2$-tensors is the symmetric gradient, which maps a $1$-form $\omega$ to the symmetric $2$-form
\[
\delta_g^* \omega_{\mu\nu} = \frac{1}{2}(\nabla_\mu \omega_\nu + \nabla_\nu \omega_\mu).
\]
We will also need the `trace reversal' operator on $2$-tensors, given by
\[
    G_g \dot{g} \coloneqq \dot{g} - \frac{1}{2}(\Tr_g \dot{g})g .
\]
Note that in the case $n=4$ we indeed have $\Tr_g \circ G_g = -\Tr_g$. Lastly we will need the Laplace--Beltrami operator, also called the (tensor) wave operator in the Lorentzian case, which acts on a $k$-tensor $u$ by
\[
\Box_g u_{\mu_1\dots\mu_k} = -\nabla^{\mu_0}\nabla_{\mu_0}u_{\mu_1\dots\mu_k}.
\]

%% file: EinsteinMaxwell.tex
\subsection{The Einstein--Maxwell equations}\label{section:EinsteinMaxwell}
    In this section, we let $M$ be a $4$-dimensional Lorentzian manifold with metric $g$. The Einstein field equations for $g$ are
    \begin{equation}\label{eq:EinsteinMaxwell}
        \Ric(g) - \frac{R}{2}g + \Lambda g = 2 T,
    \end{equation}
    where $\Ric(g)$ is the Ricci curvature tensor of $g$, $R$ is the scalar curvature, $\Lambda$ is the cosmological constant and $T$ is the energy momentum tensor. We are interested in metrics $g$ satisfying the Einstein--Maxwell equations, which are the equations above with $T$ the electromagnetic energy momentum tensor \begin{equation}
        T(g,F)_{\mu\nu} = F_{\mu\alpha}F\indices{_\nu^\alpha} - \frac{1}{4}F^{\alpha\beta}F_{\alpha\beta} g_{\mu\nu}.
    \end{equation}
    Here $F$ is the Faraday $2$-form, also called the electromagnetic tensor. 
    As this $T$ is traceless, taking traces in \eqref{eq:EinsteinMaxwell} yields $R = 4\Lambda$, so the Einstein field equations reduce to
    \begin{equation}
        \Ric(g) - \Lambda g = 2 T(g,F).
    \end{equation}
    $F$ is required to satisfy Maxwell's equations in a vacuum, which read
    \begin{gather*}
        \dl{F} = 0, \\
        \dl{\star}F = 0.
    \end{gather*}
    The Reissner--Nordström black hole, introduced after this section, possesses an electromagnetic potential $A$. This potential defines an electromagnetic tensor through $F=\dl{A}$. In coordinates
    \[
        F_{\mu\nu} = \nabla_\mu A_\nu - \nabla_\nu A_\mu,
    \]
    and we will glue on this level instead of on the level of the electromagnetic tensor. The existence of a potential automatically implies the homogenous Maxwell equations $\dl{F}= \dl{:2}A = 0$. \par
    Moreover, applying the Hodge star operator to Maxwell's second equation, we see that it is equivalent to $\delta_g \dd {A} = \star \dl{\star} F = 0$. Also note that Maxwell's equations indeed imply that the energy-momentum tensor has zero divergence.
    In summary, we define for a Lorentzian metric $g$ and a covector $A$, 
    \begin{equation}\label{eq:definitionPanddeltad}
    \begin{split}
        P(g,A) &\coloneqq 2(\Ric(g) - \Lambda g - 2T(g,\dd A))  \\
        \delta \dl{}(g,A) &\coloneqq \delta_g\dd A
    \end{split}
    \end{equation}
    A spacetime with metric $g$ and electromagnetic tensor $F=\dl{A}$ is then a solution of the Einstein--Maxwell system if and only if $P(g,A)=0$ and $\deld (g,A)=0$.

%% file: manifolds.tex
\subsection{De Sitter space and Reissner--Nordström--de Sitter black holes}
\subsubsection{De Sitter space}
    (3+1)-dimensional de Sitter space is the simplest solution to the Einstein vacuum equations with cosmological constant $\Lambda > 0$. It can be constructed as the hyperboloid $H_1 = \{(X_0,X) \subset \R^{1+4} : -X_0^2 + \abs{X}^2  = 1\}$ of radius one equipped with metric induced by the metric $3 \Lambda^{-1}(-\dl[2]{X_0} + \dl[2]{X})$.
    Viewing $\sph^3$ as a subset of $\R^4$ we may put global coordinates on it via
    \begin{align*}    
    \R_l \times \sph^3 \ni (l,\omega) &\mapsto (\sinh(l),\cosh(l)\omega),
    \end{align*}
    with the metric in these coordinates being
    \begin{equation} \label{deSitter:RadOne}
        g_{\rm{dS}} = \frac{3}{\Lambda}\big({-}\dl[2]{l} + \cosh^2{(l)}g_{\sph^3}\big).
    \end{equation}
    Setting $\tan(s/2) = \tanh(l/2)$ yields de Sitter space as $[-\pi/2,\pi/2]_s \times \sph^3$ with metric
    \[
    g_{\rm{dS}} = \frac{3}{\Lambda}\frac{-\dl[2]{s} + g_{\sph^3}}{\cos^2(s)}
    \]
    Here, $s = \pi/2$ corresponds to the future conformal boundary and $s = -\pi/2$ to the past conformal boundary. 
    It will be convenient for our investigations of the future conformal boundary to let $\tau = \cos(s)$ on the upper half $s>0$ of de Sitter space. Then $\dl[2]{\tau}= \sin^2(s)\dl[2]{s} = (1-\tau^2)\dl[2]{s}$, so the metric has the form
    \begin{equation}\label{eq:metricdSdtau}
    \frac{3}{\Lambda}\cdot\frac{-(1-\tau^2)^{-1}\dl[2]{\tau} + g_{\sph^3}}{\tau^2},
    \end{equation}
    on $[0,1)_\tau \times \sph^3$. Alternatively, we may directly take the composition of these maps on $\{(X_0,X) \in H_1 | X_0 > 0\}$, which is just
    \begin{equation} \label{deSitter:TimeTimesSphere}
        (X_0,X) \mapsto \Big(\frac{1}{\sqrt{X_0^2 + 1}}, \frac{X}{\sqrt{X_0^2 + 1}}\Big) =: (\tau, \omega).
    \end{equation}´
        
    Let us now consider coordinates on another part of de Sitter space, namely the \emph{upper half space model}. For this define the map from the upper half space to the hyperboloid of radius 1, $H_1$, via
    \begin{equation} \label{deSitter:UpperHalfSpace}
    \begin{split}
        [0,\infty)_{\tau'} \times \R^3_x &\rightarrow H_1 \\
        (\tau', x) &\mapsto \left(\frac{1 - (\tau'^2 - \abs{x}^2)}{2\tau'}, \frac{1 + (\tau'^2 - \abs{x}^2)}{2\tau'}, \frac{x}{\tau'}\right).
    \end{split}
    \end{equation}
    This map does not cover all of $H_1$, but rather only the $(X_0,X_1,X') \in H_1$ with $X_0 + X_1 > 0$. The metric takes the simple form
    \[
        \frac{3}{\Lambda}\frac{-\dl[2]{\tau'} + \dl[2]{x}}{{\tau'}^2}.
    \]
    Restricting to $\tau' \leq 1$ we have $\frac{1 - (\tau'^2 - \abs{x}^2)}{2\tau'} \geq 0$, so we may compose \cref{deSitter:TimeTimesSphere} with \cref{deSitter:UpperHalfSpace} for the map
    \begin{gather*}
        [0,1)_{\tau'} \times \R_x^3 \longrightarrow [0,1)_\tau \times \sph^3_\omega \\
        \tau = \Big( \Big(\frac{1-(\tau' - \abs{x}^2)}{2\tau'}\Big)^2 + 1\Big)^{-\frac{1}{2}}, \quad  \omega = \frac{\tau}{\tau'}\Big(\frac{1 + (\tau'^2 - \abs{x}^2)}{2}, x\Big)
    \end{gather*}
    Now put polar coordinates $x = R\omega'$, $R>0, \omega \in \sph^2$ on $\R^3_x$. In the \emph{cosmological region}, $R > \tau'$, we may let
    \[
        (t,r,\omega') = \Big(-\frac{1}{2}\sqrt{\frac{3}{\Lambda}}\log(R^2 - {\tau'}^2), \sqrt{\frac{3}{\Lambda}}\frac{R}{\tau'},\omega'\Big)
    \]
    The metric is
    \begin{equation}\label{deSitter:UpperHalfSpaceMetric}
        g_{\rm{dS}} = -\Big(\frac{\Lambda r^2}{3}-1\Big)^{-1}\dl[2]{r} + \Big(\frac{\Lambda r^2}{3}-1\Big)\dl[2]{t} + r^2g_{\sph^2}
    \end{equation}
     These coordinates are defined on $\R_t \times (\sqrt{3/\Lambda},\infty) \times \sph^2$, with $r$ taking the role of our time coordinate.
    We may compactify this by letting 
    \begin{equation}\label{deSitter:tauscoordinates}
    \tau_s = r^{-1} \in [0, \sqrt{\Lambda/3}).
    \end{equation} After this coordinate change the metric is
    \begin{equation}\label{eq:metricdStau_s}
        g_{\rm{dS}} = \frac{-(\Lambda/3 - \tau_s^2)^{-1}\dl[2]{\tau_s} + (\Lambda/3 - \tau_s^2)\dl[2]{t} + g_{\sph^2}}{\tau_s^2}
    \end{equation}
    This is a smooth $0$-metric.
    
\begin{remark}\label{remark:MetricsAtConformalBoundary}
    Let us consider the future conformal boundaries in the various coordinates. In the $\tau$ coordinate, the future conformal boundary $s = \pi/2$ is given by $\tau = 0$. There it induces the metric $h_\tau = \frac{3}{\Lambda} g_{\sph^3}$.
    On the other hand, $\tau'$ is an equivalent boundary defining function. Namely, observe that the future conformal boundary is contained in the area covered by the $(\tau',x)$ coordinates. As 
    \[
        \tau(\tau',x) = \frac{1}{\sqrt{1 + \big(\frac{1 - ({\tau'}^2 - \abs{x}^2)}{2\tau'}\big)^2}},
    \]
    it is easy to see that $\tau = 0$ is equivalent to $\tau' = 0$. We expect that the metric $h_{\tau'}$ induced by $\tau'$ at the conformal boundary is related to $h_\tau$ via $h_{\tau} = \tau^2(\tau')^{-2} h_{\tau'}$. Indeed, at $\tau' = 0$ the induced metric is $h_{\tau'} = \frac{3}{\Lambda} \dl[2]{x}$.
    To see that this is conformally related to $h_\tau$, first observe that at the conformal boundary we have
    \[
        \frac{\tau}{\tau'} = \frac{2}{1 + \abs{x}^2}
    \]
    Thus, the map between $\sph^3$ and $\R^3$ at $\tau' = 0, \tau = 0$ is just the stereographic projection. Namely, it is the map 
    \begin{equation}\label{eq:StereographicProjection}
    x \mapsto (1+\abs{x}^2)^{-1}\cdot(1-\abs{x}^2,2x),
    \end{equation}
    so the pullback of $g_{\sph^3}$ is $4(1+|x|^2)^{-2} \dl{x:2}$, whence $h_\tau = \tau^2(\tau')^{-2}h_{\tau'}$.
    Similarly, as $\tau_s = r^{-1} = \sqrt{\frac{3}{\Lambda}}\frac{\tau'}{R}$ we get $\tau_s = 0$ if and only if $\tau' = 0$.
    At the boundary we thus have
    \begin{gather*}
        \frac{\tau'}{\tau_s} = \sqrt{\frac{\Lambda}{3}}R \\
        \frac{\tau}{\tau_s} = \frac{\tau}{\tau'}\frac{\tau'}{\tau_s} = \sqrt{\frac{\Lambda}{3}}\frac{2R}{1+R^2}
    \end{gather*}
    Factor out the $3/\Lambda$ factor coming from the $\dl[2]{\tau_s}/\tau_s^2$ term at $\tau_s = 0$. The induced metric at the conformal boundary is then $\Lambda^2/9 \dl[2]{t} + \Lambda/3 g_{\sph^2}$; this will be used later in \cref{prop:CorrectionOfObstructionError}. One may verify that this is again conformally related to $h_\tau$ and $h_{\tau'}$ by the factors $\tau^2\tau_s^{-2}$ and $(\tau')^2\tau_s^{-2}$.
 \end{remark}

\subsubsection{The Reissner--Nordström--de Sitter metric}
    The Reissner--Nordström--de Sitter (RNdS) metric with mass $\mf{m} \in \R$ and charge $Q \in \R$ is given by
    \begin{equation}\label{eq:RNdS}
        g_{\mf{m},Q} = -\mu_{\mf{m},Q}(r) \dl{t:2} + \mu_{\mf{m},Q}(r)^{-1}\dl{r:2} + r^2 g_{\sph^2}  
    \end{equation}
    where $\mu_{\mf{m},Q}(r) = 1-\frac{2m}{r} + \frac{Q^2}{r^2} - \frac{\Lambda}{3}r^2$. Equipped with the vector potential $A_Q = -Q/r \dl{t}$ it is a solution to the Einstein-Maxwell system \eqref{eq:EinsteinMaxwell}. \\
    We consider \eqref{eq:RNdS} in the cosmological region $r> r_+$, where $r_+$ is the largest root of $\mu_{\mf{m},Q}(r)$ (if one exists, otherwise choose $r_+ >0$ arbitrarily). \\
    A priori, the metric seems to become singular at points where $\mu_{m,Q}(r) = 0$. These points are however only coordinate singularities. Indeed, one may choose tortoise coordinates
    \[
    t^* = t \mp r^*, \quad (r^*)' = \frac{1}{\mu}
    \]
    defined in a region above or below some root of $\mu$ (e.g. the cosmological region). In these coordinates the metric is $g_{m,Q} = -\mu_{\mf{m},Q}(r) \dl{t:2} \pm \dl{t}^*\dl{r} + r^2 g_{\sph^2}$ which is completely regular for all $r>0$. \\
    Introducing the coordinate $\tau_s \coloneqq 1/r$ makes 
    \begin{equation}\label{eq:RNdSintaus}
    g_{\mf{m},Q} = \frac{\lambda_{\mf{m},Q}(\tau_s)^{-1}\dl{}\tau_s^2 -  \lambda_{\mf{m},Q}(\tau_s)\dl{t:2} + g_{\sph^2}}{\tau_s^2}
    \end{equation}
    a smooth $0$-metric on $[0,1/r_+) \times \R_t \times \sph^2_\omega$, where $\lambda_{\mf{m},Q}(\tau_s) = \tau_s^2 - 2\mf{m}\tau_s^3 + Q^2\tau_s^4 - \Lambda/3$. The vector potential and electromagnetic tensor are
    \begin{equation}\label{eq:electromagneticpotential}
        A_Q = -Q\tau_s^2 \frac{\dl{t}}{\tau_s}, \quad F_Q = -Q\tau_s^2 \frac{\dl{}\tau_s}{\tau_s} \wedge \frac{\dl{t}}{\tau_s}.
    \end{equation}\\

%% file: LinearizationOfEinsteinMaxwell.tex
 \subsection{The linearization of the Einstein--Maxwell equations}
    Here we study the leading order behavior of the Einstein--Maxwell equations by finding their linearizations and the indicial families of these. Let $L_{0,g,A}$ and $K_{0,g,A}$ be the linearizations in both arguments $g,A$ of $P(g,A)$ and $\deld(g,A)$, respectively. Let us also define 
    \[
    P_0(g) = P(g,0) = 2(\Ric(g) - \Lambda g),
    \]
    and let $L_{0,g}$ be its linearization at $g$.\\
    In \cite{DeTurck81, GrahamLee, blackholegluing}, it was calculated that
    \begin{equation}
        L_{0,g} = \Box_g - 2\delta_g^*\delta_g\sfG_g + 2\sR_g - 2\Lambda,
    \end{equation}
    where 
    \begin{equation}
        \sR_g(u)_{\mu\nu} = \rm{Riem}\indices{^\alpha_\mu_\nu^\beta}u_{\alpha\beta} + \frac{1}{2}(\Ric\indices{_\mu^\alpha}u_{\alpha\nu} + \Ric\indices{_\nu^\alpha}u_{\alpha\mu}).
    \end{equation}
    Here $\rm{Riem}$ is the Riemann curvature tensor of $g$. To further examine this we follow \cite{blackholegluing} by introducing certain bundle splittings and then calculating the components of the differential operators in these splittings.\\
    Let $M = [0,1)_\tau \times X$, where $X$ is a 3-dimensional manifold. We let 
\begin{equation}
    \label{eq:def0-frame}
    \begin{alignedat}{2}
    &e^0 = \frac{\dl{\tau}}{\tau}, \quad& &e^i = \frac{\dl{x^i}}{\tau}, \\
    &e_0 = \tau \del_\tau, \quad& &e_i = \tau \del_{x^i},
    \end{alignedat}
\end{equation}
be frames for the $0$-cotangent, respectively $0$-tangent, bundle for some chart $x$ of $X$.
\ 
    Split $\nT^*M$ and $S^2 \nT^*M$ according to
\begin{equation}
\begin{split}\label{eq:splittingsymmetric}
    \nT^*M &= \R e^0 \oplus \tau^{-1}T^*X, \\
    S^2 \nT^*M &= \R (e^0)^2 \oplus (2e^0 \otimes_s\tau^{-1}T^*X) \oplus \tau^{-2}S^2 T^*X.
\end{split}
\end{equation}
For a 1-form $\alpha$, we write $\alpha_N$ for its normal part and $\alpha_T$ for its tangential part in this splitting. If we take $h$ to be a Riemannian metric on $X$, the splitting for $S^2 \nT^*M$ may further be refined to
\begin{equation}\label{eq:refinedsplitting}
    S^2 \nT^*M = \R (e^0)^2 \oplus (2e^0 \otimes_s \tau^{-1}T^*X) \oplus \R \tau^{-2}h \oplus \tau^{-2}\Ker \Tr_h.
\end{equation}
That is, we split a symmetric $2$-tensor $t$ on $X$ into its traceless part $t-\frac{1}{3} h \Tr_h t$ and its pure trace part $\frac{1}{3} h \Tr_h t$. We use the notation $t_{NN}, t_{NT}, t_{TT1}, t_{TT0} \in \R$ for normal-normal, normal-tangential, tangential-tangential-pure trace, tangential-tangential-trace-free parts in \eqref{eq:refinedsplitting}.\par
As we will be working with exterior derivatives in Maxwell's equations, we split the antisymmetric $(0,2)$ $0$-tensor bundle $\Lambda^2 \nT^*M$ via
\begin{equation}\label{eq:splittingantisymmetric}
    \Lambda^2 \nT^*M = (e^0 \wedge \tau^{-1}T^*X) \oplus \tau^{-2}\Lambda^2 T^*X.
\end{equation}

    Then, if $M$ is equipped with the product metric
    \begin{equation}\label{eq:baseformproductmetric}
        g = \frac{3}{\Lambda}\frac{-\dl{\tau}^2 + h(x,\dl{x})}{\tau^2},
    \end{equation}
    we obtain matrix representations of 0-differential operators in the splittings \eqref{eq:splittingsymmetric}. We cite:
    \begin{lemma}[{\cite[Lemma 2.4.]{blackholegluing}}]\label{lemma:splittingsDiffOpOnSymmetric}
  \[
    \delta_g^* = \begin{pmatrix} e_0 & 0 \\ \frac{1}{2}\tau \rm{d}_X & \frac{1}{2}(1+e_0) \\ h & \tau\delta_h^* \end{pmatrix}, \quad
    3\Lambda^{-1}\delta_g = \begin{pmatrix} e_0-3 & \tau\delta_h & -\Tr_h \\ 0 & e_0-4 & \tau\delta_h \end{pmatrix}
  \]
  and, as operators on symmetric 2-tensors,
  \[
    \sfG_g = \begin{pmatrix} \frac{1}{2} & 0 & \frac{1}{2}\Tr_h \\ 0 & 1 & 0 \\ \frac{1}{2} h & 0 & \sfG_h \end{pmatrix}, \quad
    3\Lambda^{-1}\Box_g = e_0^2 - 3 e_0 + \tau^2\Delta_h + \begin{pmatrix} -6 & 4\tau\delta_h & -2\Tr_h \\ -2\tau \rm{d}_X & -6 & 2\tau\delta_h \\ -2 h & -4\tau\delta_h^* & -2 \end{pmatrix}.
  \]
  Finally, if $R_{\kappa\lambda\mu\nu}$ and $\Ric_{\mu\nu}$ denote the Riemann curvature tensor and Ricci tensor of $g$, then the operator $\sR_g(u)_{\kappa\mu}=R^\nu{}_{\kappa\mu}{}^\rho u_{\nu\rho}+\frac{1}{2}(\Ric_\kappa{}^\nu u_{\nu\mu}+\Ric_\mu{}^\nu u_{\kappa\nu})$ is equal to
  \[
    3\Lambda^{-1}\sR_g = \begin{pmatrix} 3 & 0 & \Tr_h \\ 0 & 4 & 0 \\ h & 0 & 4-h\Tr_h \end{pmatrix} + \tau^2\begin{pmatrix} 0 & 0 & 0 \\ 0 & \frac{1}{2}\Ric(h) & 0 \\ 0 & 0 & \sR_h \end{pmatrix}.
  \]
    \end{lemma}

In the refined splitting \eqref{eq:refinedsplitting} the linearization $L_{0,g}$ of $P_0$ therefore is, modulo $\tau^2 \Diff_0^2$ terms
\begin{align*}
    3\Lambda^{-1}L_{0,g} \equiv \begin{pmatrix}
3(e_0 - 2) & 2\tau(1-e_0)\delta_h & 3e_0(2-e_0) & 0 \\
2\tau \rm{d}_X& 0 & -2\tau e_0 \rm{d}_X & -\tau e_0 \delta_h\\
6-e_0 & \frac23 \tau (e_0 - 5) \delta_h & e_0(e_0-6)& 0 \\
0 & 2\tau(2-e_0) \delta_{h,0}^* & 0 & e_0(e_0-3)
\end{pmatrix}
\end{align*}
    
    We now additionally calculate the linearization of $-4T$ in the arguments $g$ and $A$:
    \begin{equation}\label{eq:CalcLinearization4T}
        \begin{split}
        D_g[-4T(\cdot, A)](\dot{g})_{\mu\nu} &= 4 F\indices{_\mu^\alpha} F\indices{_\nu^\beta} \dot{g}_{\alpha\beta} + F^{\alpha\beta}F_{\alpha\beta} \dot{g}_{\mu\nu} - 2 F^{\alpha\lambda}F\indices{_{\alpha}^\rho} \dot{g}_{\lambda\rho} g_{\mu\nu}  \\
        D_A [-4T(g,\cdot)](\dot{A})_{\mu\nu} &= -4F\indices{_\mu^\alpha}\dot{F}\indices{_\nu_\alpha} - 4F\indices{_\nu^\alpha}\dot{F}\indices{_\mu_\alpha} + 2\dot{F}^{\alpha\beta}F_{\alpha\beta}g_{\mu\nu},
        \end{split}
    \end{equation}
        where $\dot{F} = d\dot{A}$. 
        Similarly, the linearizations of $\deld$ at $g$ and $A$ are
        \begin{gather}
            D_g [\deld(\cdot,A)](\dot{g})_\mu = F\indices{_\mu^\alpha}\Big(\frac{1}{2}\nabla_\alpha \Tr_g \dot{g} - \nabla_\beta \dot{g}\indices{_\alpha^\beta}\Big) + F^{\alpha\beta} \nabla_\mu \dot{g}_{\alpha\beta} + \dot{g}^{\alpha\beta}\nabla_\beta F_{\alpha\mu} \\
            D_A[\deld(g,\cdot)](\dot{A})_\mu = \delta_g \rm{d}(\dot{A})
        \end{gather}
        In case $g,A$ come from a de Sitter background, so $g=g_{dS}, A=0$, note that all these linearizations, except for $D_A(\delta_g \rm{d})$, vanish.

\begin{lemma}\label{lemma:splittingdeltagd}
In the splitting \eqref{eq:splittingantisymmetric} we have
\[
 3\Lambda^{-1}\delta_g = \begin{pmatrix}
        -\tau \delta_h & 0 \\ e_0-2 & \tau\delta_h \end{pmatrix}, \quad d = \begin{pmatrix}
-\tau \rm{d}_X & e_0-1 \\ 0 & \tau \rm{d}_X
    \end{pmatrix}
\]
so
\[
3\Lambda^{-1}\delta_gd = \begin{pmatrix}
    \tau^2 \delta_h\rm{d}_X & \tau(1-e_0)\delta_h\\
    \tau(1-e_0)\rm{d}_X & (e_0-2)(e_0-1) + \tau^2 \delta_h \rm{d}_X
\end{pmatrix}
\]
\end{lemma} 
\begin{proof}
    Take a local coordinate chart $(x^1, \dots, x^n)$ on $X$. Then for small $\tau$, $(\tau,x)$ is a boundary chart which gives us a frame $e^0 = \frac{\rm{d}\tau}{\tau}, e^i = \frac{\rm{d}x_i}{\tau}$ for the $0$-cotangent bundle and a frame $e_0 = \tau \del_\tau, e_i = \tau \del_{x^i}$ for the $0$-tangent bundle, as in \eqref{eq:def0-frame}. We calculate that in this frame
    \begin{equation}
        \begin{split}
            \nabla_{e_0} e^\mu = 0, \quad \nabla_{e_i}e^0 = h_{ik} e^k, \quad \nabla_{e_i}e^k = \delta_i^k e^0 - \tau \Gamma(h)_{ij}^k e^j.
        \end{split}
    \end{equation}
    Here $\delta_i^k$ is the Kronecker delta. For demonstration we calculate the bottom left component of $3 \Lambda^{-1} \delta_g$; the others are analogous. Let $\alpha$ be a section of $T^*X$. To show that the bottom left component is $e_0 - 2$, we need to show that
    \[
    \frac{3}{\Lambda}\delta_g(e^0 \wedge \frac{\alpha}{\tau})_{\overline{i}} = (e_0 - 2)\alpha_i,
    \]
    where we denote with an overline components in the $0$-frames, so $g^{\overline{\mu}\overline{\nu}} = \tau^2 g^{\mu\nu}$ etc.
    But this is indeed the case, as
    \begin{align*}
    \frac{3}{\Lambda}\delta_g\Big(e^0 \wedge \tau^{-1}\alpha\Big)_{\overline{i}} &\begin{aligned}[t]
    = - \frac{3}{\Lambda} \big( &\nabla^{\overline{\mu}}(e^0)_{\overline{\mu}} (\tau^{-1}\alpha)_{\overline{i}} + (e^0)_{\overline{\mu}}\nabla^{\overline{\mu}}(\tau^{-1}\alpha)_{\overline{i}} \\ 
    &- \nabla^{\overline{\mu}}(\tau^{-1}\alpha)_{\overline{\mu}} (e^0)_{\overline{i}} - (\tau^{-1}\alpha)_{\overline{\mu}}\nabla^{\overline{\mu}}(e^0)_{\overline{i}}\big) \end{aligned}\\
    &= -3 \alpha_i + e_0 \alpha_i +0 + \alpha_i \\
    &= (e_0 - 2)\alpha_i,
    \end{align*}

    where we have used $\nabla^{\overline{\mu}}(e^0)_{\overline{\mu}} = \Lambda$ and $\nabla^{\overline{k}}(e^0)_{\overline{i}} = \delta_i^k$, $\nabla^{\overline{\tau}}(e^0)_{\overline{i}} = 0$.
\end{proof}
The previous lemmas now give the following matrix representations of the indicial families of $P$ and $\deld$:
\begin{equation}\label{eq:linearizationP}
    \begin{split}
    3\Lambda^{-1}I(L_{0,g,0},\lambda) &= \begin{pmatrix}
        3(\lambda - 2) & 0 & 3\lambda(2-\lambda) & 0 & 0 & 0 \\
        0 & 0 & 0 & 0 & 0 & 0 \\
        6-\lambda & 0 & \lambda(\lambda-6) & 0 & 0 & 0 \\
        0 & 0 & 0 & \lambda(\lambda-3) & 0 & 0
        \end{pmatrix} 
        \\
       3\Lambda^{-1}I(L_{0,g,0}[\tau],\lambda) &= \begin{pmatrix}
        0 & 2(1-\lambda)\delta_h & 0 & 0 & 0 & 0 \\
        2\rm{d}_X & 0 & -2\lambda \rm{d}_X & -\lambda \delta_h & 0 & 0 \\
        0 & \frac23 (\lambda -5)\delta_h & 0 & 0 & 0 & 0\\
        0 & (4-2\lambda)\delta_{h,0}^* & 0 & 0 & 0 & 0
        \end{pmatrix} 
    \end{split}
\end{equation} 
\begin{equation}  \label{eq:linearizationdeltad}
    \begin{split}
        3\Lambda^{-1}I(K_{0,g,0},\lambda) &= 
        \begin{pmatrix}
        0 & 0 & 0 & 0 & 0 & 0 \\
        0 & 0 & 0 & 0 & 0 & (\lambda-2)(\lambda-1) \\
    \end{pmatrix}\\
        3\Lambda^{-1}I(K_{0,g,0}[\tau],\lambda) &=
        \begin{pmatrix}
        0 & 0 & 0 & 0 & 0 & (1-\lambda)\delta_h  \\
        0 & 0 & 0 & 0 & (1-\lambda)\rm{d}_X & 0\\
    \end{pmatrix}
    \end{split}
\end{equation}
Here the first four columns describe the indicial families' action on symmetric $2$-tensors, the last two on potentials. Also note that $L_{0,g}$ consists of the first four columns of $L_{0,g,0}$.\\
We will use two types of argument various times throughout the gluing construction. The first is that the linearization is stable up to order $\tau^m$ under disturbances of order $\tau^{m+1}$:
\begin{lemma}[{\cite[Lemma 2.6]{blackholegluing}}]\label{lemma:StabilityOfLinearization}
    Let $g \in \mc{C}^\infty(M; S^2 \nT^*M)$ be a smooth Lorentzian $0$-metric and $A \in \mc{C}^\infty(M; \nT^*M)$ a smooth potential. Suppose we have perturbations $\Tilde{g} \in \tau^{m_1}\mc{C}^\infty(M;S^2 \nT^*M)$ and $\Tilde{A} \in \tau^{m_2}\mc{C}^\infty(M; \nT^*M)$ for some $m_1,m_2 \in \N$. Then
    \begin{gather*}
        L_{0,g+\Tilde{g}, A+\Tilde{A}} - L_{0,g,A} \in \tau^{\min \{m_1,m_2\}} \mc{C}^\infty, \\
        K_{0,g+\Tilde{g}, A+\Tilde{A}} - K_{0,g,A} \in \tau^{\min \{m_1,m_2\}} \mc{C}^\infty.
    \end{gather*}
    In particular, the indicial families $I(L_{0,g+\Tilde{g}, A+\Tilde{A}}[\tau^k],\lambda), I(L_{0,g+\Tilde{g}, A+\Tilde{A}}[\tau^k],\lambda)$ are independent of the perturbations for $k < \min\{m_1,m_2\}$.
\end{lemma}
\begin{remark}
    We have slightly abused notation here and only written $\mc{C}^\infty$ instead of writing out the specific vector bundles of which the maps are sections of. We will do this a lot more; whenever we write $\mc{C}^\infty$ we mean $\mc{C}^\infty$ as sections of the respective $0$-vector bundles.
\end{remark}
\begin{proof}
    We prove it for $L$ with the proof for $K$ being analogous. Start with the expansion
    \[
    (g+\Tilde{g})^{-1} = g^{-1}(id - \Tilde{g}g^{-1}) + \mc{O}(\tau^{2m_1}).
    \]
    This implies $(g+\Tilde{g})^{-1} - g^{-1} \in \tau^{m_1}\mc{C}^\infty$, and using this we may observe that $\Gamma(g+\Tilde{g}) - \Gamma(g) \in \tau^{m_1}\mc{C}^\infty$. Therefore also $\delta_{g+\Tilde{g}}-\delta_g \in \tau^{m_1}\mc{C}^\infty$ and similarly for the other differential operators appearing in $L$.
\end{proof}
According to \eqref{eq:metricdSdtau}, \eqref{eq:metricdStau_s}, the de Sitter metric is, up to order $\tau^2\mc{C}^\infty$, respectively $\tau_s^2\mc{C}^\infty$ terms, in the form of a product metric as in \eqref{eq:baseformproductmetric}. We may therefore use the lemma to see that \eqref{eq:linearizationP} and \eqref{eq:linearizationdeltad} indeed agree with the leading and subleading order indicial families of the linearizations of $P$ and $\deld$ at the de Sitter metric. \\
The second argument we often make use of is that we may use the linearization to calculate lower order error terms of differential operators acting on perturbed elements:
\begin{lemma}[{\cite[Lemma 2.7]{blackholegluing}}]\label{lemma:DifferentialOperatorsUpToLinearization}
    With $P$ and $\deld$ defined in \eqref{eq:definitionPanddeltad}, let $g$ be a smooth Lorentzian $0$-metric and $A \in \mc{C}^\infty(M; \nT^*M)$ a smooth potential. Suppose we have perturbations $\Tilde{g} \in \tau^{m_1}\mc{C}^\infty(M;S^2 \nT^*M)$ and $\Tilde{A} \in \tau^{m_2}\mc{C}^\infty(M; \nT^*M)$ for some $m_1,m_2 \in \N$. Then
    \begin{gather*}
        P(g + \Tilde{g}, A + \Tilde{A}) - P(g, A) - L_{0,g,A}(\Tilde{g},\Tilde{A}) \in \tau^{2\min \{m_1,m_2\}}\mc{C}^\infty \\
        \deld(g + \Tilde{g}, A+ \Tilde{A}) - \deld(g,A) - K_{0,g,A}(\Tilde{g},\Tilde{A}) \in \tau^{\min \{2m_1,m_1 + m_2\}}\mc{C}^\infty
    \end{gather*}
\end{lemma}
\begin{proof}
    This follows from the same considerations as in the last lemma. We first prove it for $P$. By subtracting $P(g,A)$ and $L_{0,g,A}(\Tilde{g},\Tilde{A})$ we exactly cancel out the zeroth and first order terms in $\Tilde{g}$ and $\Tilde{A}$. Consequently, any remaining terms must be at least square in the disturbances. Since $0$-derivatives preserve the $\tau^{m_1}$ and $\tau^{m_2}$ decay rates respectively, this finishes the proof for $P$. For $\deld$ note that it is linear in $A$. Therefore, the only second-order terms in $\Tilde{g}$ and $\Tilde{A}$ are a product of $\Tilde{A}$ and $\Tilde{g}$, or quadratic terms in $\Tilde{g}$. This yields the stated congruence for $\deld$.
\end{proof}

%% file: gluingexplanation.tex
\subsection{Naive gluing and the balance condition}
    We are now prepared to discuss the gluing in more detail. To start, by comparing \eqref{eq:RNdSintaus} and \eqref{eq:metricdStau_s} one sees that
    \begin{equation} \label{eq:gdS-gRNdSistau3}  
    g_{\rm{dS}} - g_{\mathfrak{m},Q} \in \tau_s^3\mc{C}^\infty(M; S^2 \nT^*M)
    \end{equation}
    on their common domain of definition $[0, \min\{\sqrt{\Lambda/3}, r_+^{-1}\})_{\tau_s} \times \R_t \times \sph^2$. Additionally, we have
    \begin{equation}\label{eq:AdS-A_Qistau2}
        A_Q \in \tau_s^2 \mc{C}^\infty(M; \nT^*M).
    \end{equation}
    
    At the conformal boundary, $\{\tau_s = 0\}$, using the upper half space coordinates \eqref{deSitter:UpperHalfSpace} (and spherical coordinates on $\R^3$) we find that $R = e^{-t\sqrt{3/\Lambda}}$. Therefore, $R = 0$ corresponds to $t \rightarrow \infty$. This observation allows us to interpret the $g_{\mathfrak{m},Q}$ metric as gluing a de Sitter black hole into de Sitter space at the point $\tau_s = 0, R= 0$, which is just the point $p_0 = (1,0,0,0) \in \sph^3$ in the coordinates \eqref{deSitter:TimeTimesSphere}. If instead we want to glue the metric into any other point $p \in \sph^3$, we may choose a rotation $T \in SO(4)$ which maps $p$ to $p_0$ and then pull back $g_{\mf{m},Q}$ and $A_Q$ along $T$. Let us call these metrics $g_{p,\mf{m},Q}$ and potentials $A_{p,Q}$.

    Knowing this, our goal is to show the following theorem:
    \begin{theorem}\label{theorem:BlackHoleGluing}
        Let $N \in \N$ and let $p_i \in \sph^3, \mf{m}_i,Q_i \in \R$ for $1 \leq i \leq N$ satisfy the \emph{charge balance condition} 
        \begin{equation}\label{eq:ChargeBalanceCondition}
         \sum_{i=1}^N Q_i = 0
        \end{equation}
        and the \emph{mass balance condition}
        \begin{equation}\label{eq:MassBalanceCondition}
        \sum_{i=1}^N \mf{m}_i p_i = 0 \in \R^4.
        \end{equation}
        Let $V_{p_i} \subset \sph^3 = \del M$ be a neighborhood of $p_i$ and assume that $\overline{V_{p_i}} \cap \overline{V_{p_j}} = \emptyset$ for all $i \neq j$. Then there exists a neighborhood $U$ of $\del M \setminus \{p_1, \dots, p_N\}$, a Lorentzian $0$-metric $g \in \mc{C}^\infty(U;S^2 \nT_U^*M)$, and a vector potential $A \in \mc{C}^\infty(U;\nT^*M)$ with the following properties:
        \begin{enumerate}
            \item $g$ and $A$ satisfy the Einstein--Maxwell equations, so $\Ric(g) - \Lambda g = 2T$ and $\delta_g \rm{d} A = 0$,
            \item in a punctured neighborhood of $V_{p_i}$ in $M$ we have $g = g_{p_i,\mf{m_i},Q_i}$, $A = A_{p_i,Q_i}$,
            \item $g-g_{dS} \in \tau^3\mc{C}^\infty(U; S^2 \nT_U^*M)$,
            \item $A \in \tau^2 \mc{C}^\infty(U; \nT_U^*M)$.
        \end{enumerate}
    \end{theorem}

    The procedure for proving \cref{theorem:BlackHoleGluing} will be as follows. We first choose cutoff functions $\chi_i$, which are equal to $1$ near $\overline{V_{p_i}}$ and have mutually disjoint supports. We then naively glue the black holes into de Sitter space by defining the metric
    \begin{equation}\label{eq:DefinitionMetricg_(3)} 
    g_{(3)} \coloneqq \chi_0 g_{dS} + \sum_{i=1}^N \chi_i g_{p_i,\mf{m}_i,Q_i},
    \end{equation}
    and vector potential
    \begin{equation}\label{eq:DefinitionPotentialA_(2)}
    A_{(2)} \coloneqq \sum_{i=1}^N \chi_i A_{p_i, Q_i},
    \end{equation}
    where $\chi_0 = 1 - \sum_{i=1}^N \chi_i$. As before, we have $g_{(3)} - g_{dS} \in \tau_s^3 \mc{C}^\infty$ and $A_{(2)} \in \tau_s^2 \mc{C}^\infty$.
    We will then find the correction in three key steps.
    \begin{enumerate}
        \item The first step is to increase the order in $\tau$ up to which the Einstein--Maxwell equations are satisfied by one. To accomplish this, we start by calculating that the error to $\deld(g_{(3)}, A_{(2)})$ is in $\tau^3\mc{C}^\infty$. We then find an improved potential $A_{(3)}$ such that $A_{(3)}- A_{(2)} \in \tau^2\mc{C}^\infty$ has the desired support restrictions by solving an underdetermined divergence equation on $1$-forms using cohomological methods. Existence of such a solution is guaranteed by the charge balance condition. \\
        Similarly, the resulting error to $P(g_{(3)}, A_{(3)})$ will be in $\tau^4 \mc{C}^\infty$ and finding an improved metric $g_{(4)}$ which only changes $g_{(3)}$ away from the $\overline{V_{p_i}}$ requires solving the same underdetermined divergence equation for symmetric $2$-tensors as in \cite{blackholegluing}.
        \item We then improve the newly found $g_{(4)}, A_{(3)}$ to a formal solution of the Einstein--Maxwell system using an iteration scheme. Namely, we will alternate between increasing the order of $\deld$ by one power of $\tau$ by finding a correction to the potential, and increasing the order of $P$ by finding a correction to the metric. We do not fix a gauge yet; instead, an exactness argument using the second Bianchi identity on $P$ and the vanishing of $\delta^2$ on $\deld$ secures the existence of the correction terms. The formal solutions $g_{(\infty)}$ and $A_{(\infty)}$ are then constructed using Borel's lemma. These satisfy $P(g_{(\infty)},A_{(\infty)}) \in \tau^\infty \mc{C}^\infty$, $\deld(g_{(\infty)},A_{(\infty)}) \in \tau^\infty \mc{C}^\infty$.
        \item Lastly, a DeTurck gauge for $P$ and a Lorenz gauge for $\deld$ with background metric $g_{(\infty)}$ and background potential $A_{(\infty)}$ converts the Einstein--Maxwell equations into quasilinear wave operators. Solutions $g,A$ to the resulting system will be found by solving the wave equations backwards from $\del M$. They satisfy $g-g_{(\infty)} \in \tau^\infty \mc{C}^\infty$, $A-A_{(\infty)} \in \tau^\infty \mc{C}^\infty$ and vanish near the $p_i$. The added gauge terms will then themselves solve a quasilinear wave equation and by uniqueness vanish. The constructed solutions $g,A$ will therefore solve the Einstein--Maxwell system. 
    \end{enumerate}

%% file: LeadingOrderCorrectionObstruction.tex
\subsection{Leading order correction - the obstruction}
In this section we calculate the error terms of the Einstein--Maxwell equations for a naive gluing of RNdS black holes into de Sitter space. We show the following proposition:
\begin{proposition}\label{prop:CorrectionOfObstructionError}
    Let $g_{(3)}$ and $A_{(2)}$ be the naively glued metric and potential defined in \eqref{eq:DefinitionMetricg_(3)} and \eqref{eq:DefinitionPotentialA_(2)} and assume that they satisfy the requirements of \cref{theorem:BlackHoleGluing}. Then there exist $g_{(4)} \in \mc{C}^\infty(M; S^2 \nT^*M)$ and $A_{(3)} \in \tau^2\mc{C}^\infty(M; \nT^*M)$ such that 
    \begin{enumerate}
        \item $g_{(4)}-g_{(3)} \in \tau^3\mc{C}^\infty$ and $g_{(4)} = g_{(3)}$ near $\bigcup_{i=1}^N \overline{V_{p_i}}$;
        \item $A_{(3)}-A_{(2)} \in \tau^2\mc{C}^\infty$ and $A_{(3)} = A_{(2)}$ near $\bigcup_{i=1}^N \overline{V_{p_i}}$;
        \item 
        $P(g_{(4)}, A_{(3)}) \in \tau^5\mc{C}^\infty$;
        \item $\deld(g_{(4)}, A_{(3)}) \in \tau^4\mc{C}^\infty$.
    \end{enumerate}
\end{proposition}
In other words, we may find corrections to the naively glued metric and potential that do not affect the black hole regions and satisfy the Einstein--Maxwell equations to one order higher. \par
We first consider the case of gluing a single black hole. Here, there will be a non-trivial obstruction to the gluing, from which we derive the charge balance condition later. We also start with a cutoff function depending only on the coordinate $t$. Thus let $\chi \in \mc{C}^\infty(\R_t)$ be some function which is $1$ for $t \gg 0$ and $0$ for $t \ll 0$. As in \eqref{eq:DefinitionMetricg_(3)} and \eqref{eq:DefinitionPotentialA_(2)} define the naively glued metric and electromagnetic potential
    \begin{equation}\label{eq:singleblackholegluing}
        g_{(3)} \coloneqq \chi(t)g_{\mf{m},Q} + (1-\chi(t))g_{\rm{dS}} \quad A_{(2)} = \chi(t)A_Q.
    \end{equation}
We may write $\gamma \coloneqq g_{\mf{m},Q} - g_{\rm{dS}} = \tau_s^3 \gamma_3 + \tau_s^4 \gamma_4 + \mc{O}(\tau_s^5)$ with $\gamma_i$ independent of $\tau_s$. We then have
\[
    g_{(3)} = g_{\rm{dS}} + \chi(t)\gamma
\]
and we calculate that in the refined splitting \eqref{eq:refinedsplitting}, with boundary defining function $\tau_s$, $e^0 = \dl{\tau_s}/\tau_s$ and spatial metric $h_s = (\Lambda^2/9)\dl{t}^2 + (\Lambda/3)g_{\sph^2}$,
\begin{equation}\label{eq:gamma3andgamma4}
    \begin{split}
        \gamma_3 &= 2\mf{m}\Big(\frac{9}{\Lambda^2}, 0, \frac{3}{\Lambda^2}, \frac{2}{3}\dl{t:2} - \frac{1}{\Lambda}g_{\sph^2}\Big) \\
        \gamma_4 &= -Q^2\Big(\frac{9}{\Lambda^2}, 0, \frac{3}{\Lambda^2}, \frac{2}{3}\dl{t:2} - \frac{1}{\Lambda}g_{\sph^2}\Big)
    \end{split}
\end{equation}
Let us similarly (for consistency in notation) write $A_{(2)} = \chi(t) A_{p,Q} = \chi(t) \alpha$, where $\alpha = A_{p,Q} = \tau_s^2\alpha_2 = \tau_s^2(0, -Q\dl{t})$ in the splitting \eqref{eq:splittingsymmetric}.
\begin{lemma}\label{lemma:ErrorDeltaD}
    We have
    \begin{align*}
    \deld(g_{(3)},A_{(2)}) &\equiv \tau_s^3 \ \rm{Err}_{\deld,s} \mod \tau_s^4 \mc{C}^\infty,
    \end{align*}
    where $\rm{Err}_{\deld,s} = -\frac{3}{\Lambda}Q \chi'(t) \frac{\dl{\tau_s}}{\tau_s}$. 
\end{lemma}
\begin{proof}
    Since $\delta_{g_{(3)}}$ and $\rm{d}$ are $0$-differential operators, and $A_{(2)} = \mc{O}(\tau_s^2)$, we know that the error is at least $\mc{O}(\tau_s^2)$. Now because $g_{\rm{dS}}$ with zero electromagnetic potential solves Maxwell's equations $\deld(g_{\rm{ds}},0) = 0$, we have according to \cref{lemma:DifferentialOperatorsUpToLinearization}
    \begin{align*}
    \deld(g_{(3)},A_{(2)}) &\equiv K_{0,g_{dS},0}(\chi\gamma, \chi \alpha) \mod \tau_s^5 \mc{C}^\infty
    \end{align*}
    We may calculate the terms for the different orders of $\tau_s$ using the indicial families of $K_{0,g,A}$ given in \eqref{eq:linearizationdeltad}, with spatial metric $h_s = \frac{\Lambda^2}{9} \dl{t:2} + \frac{\Lambda}{3}g_{\sph^2}$. First, observe that for $\chi \equiv 1$ the left side of the above equation vanishes (as the RNdS metric and electromagnetic potential are a solution of the Einstein Maxwell equations). Thus,
    \begin{gather*}
    I(K_{0,g_{\rm{dS}},0},2)(0,\alpha_2) = 0 \\
    I(K_{0,g_{\rm{dS}},0}[\tau_s],2)(0,\alpha_2) + I(K_{0,g_{\rm{dS}},0},3)(\gamma_3,0) = 0
    \end{gather*}
    which is also simple to check directly.
    The order $\tau_s^2$ error is hence
    \begin{align*}
    I(K_{0,g_{\rm{dS}},0},2)(0,\chi \alpha_2) &= \chi I(K_{0,g_{\rm{dS}},0},2)(0,\alpha_2) \\ &=0.
    \end{align*}
     The $\tau_s^3$ error is
    \begin{align*}
    I(K_{0,g_{\rm{dS}},0}[\tau_s],2)(0,\chi \alpha_2) + &I(K_{0,g_{\rm{dS}},0},3)(\chi\gamma_3,0) \\
    &= I(K_{0,g_{\rm{dS}},0}[\tau_s],2)(0,\chi \alpha_2) + \chi  I(K_{0,g_{\rm{dS}},0},3)(\gamma_3,0) \\
    &=I(K_{0,g_{\rm{dS}},0}[\tau_s],2)(0,\chi \alpha_2) - \chi I(K_{0,g_{\rm{dS}},0}[\tau_s],2)(0,\alpha_2) \\
    &= [I(K_{0,g_{\rm{dS}},0}[\tau_s],2), \chi](0,\alpha_2) \\
    &= \frac{\Lambda}{3}(-[\delta_{h_s},\chi](-Q) \dl{t},0) \\
    &= \Big(-\frac{3}{\Lambda} Q \chi'(t) , 0\Big)
    \end{align*}
    where we have used that $-[\delta_{h_s}, \chi] = \iota_{({}^{h_s}\nabla \chi)} = 9 \Lambda^{-2} \chi'(t)\iota_{\del_t}$, with $\iota$ denoting contraction, in the last step.
\end{proof}
Because the normal part of $I(K_{0,g_{\rm{dS}},0},3)$ vanishes, we cannot solve away this error by adding an order $\tau_s^3$ term. Instead, considering the form of the subleading indicial family of $\deld$, we want to find a 1-form $\tilde{\alpha} \in \mc{C}^\infty(\del M, T^*\del M)$ which solves $-\delta_{h_s} \tilde{\alpha} = -(\rm{Err}_{\deld,s})_N$ and vanishes near $p_0$. Because $I(K_{0,g_{dS},0},2) = 0$, this does not produce an order $\tau^2$ error. As $\delta_{h_s} = \star_{h_s} \rm{d} \star_{h_s}$, this is equivalent to finding a $2$-form $\omega = \star_{h_s} \tilde{\alpha}$ supported away from $p_0$ which satisfies
\[
d \omega = - \star_{h_s} (\rm{Err}_{\deld,s})_N = - (\rm{Err}_{\deld,s})_N \rm{d}h_s.
\]
A necessary condition for such an $\omega$ to exist is that
\begin{equation}\label{eq:chargeconditionOne}
    \begin{split}
    0 &= \int_{\del M} (\rm{Err}_{\deld,s})_N \rm{d}h_s \\
    &= -\frac{3}{\Lambda}Q \int_{\R_t \times \sph^2} \chi'(t) \rm{d}h_s \\
    &= -Q \frac{\Lambda}{3}\rm{vol}(\sph^2)
    \end{split}
\end{equation}

We note here that this condition is conformally invariant; indeed, passing to the boundary defining coordinate $\tau$ the boundary metric is $g_{\sph^3}$ and we have
$\rm{d}h_s = \tau^3\tau_s^3\rm{d}g_{\sph^3}$. On the other hand the $\tau^3$ error is 
\begin{align*}
(\rm{Err}_{\deld})_N &= \tau^{-3}\deld(g_{(3)},A_{(2)})(\tau\del_\tau)\vert_{\del M} \\
&= \tau^{-3}\tau_s^{3}\tau_s^{-3}\deld(g_{(3)},A_{(2)})(\tau_s\del_{\tau_s})\vert_{\del M} \\
&= \tau^{-3}\tau_s^{3}(\rm{Err}_{\deld, s})_N.
\end{align*}
As a result of \eqref{eq:chargeconditionOne}, gluing a single RNdS black hole into de Sitter space is not possible in this manner. For multiple black holes, the situation is different, however. We may proceed with the gluing as long as the black holes satisfy the charge balance condition:
\begin{lemma}\label{lemma:chargebalancecondition}
    Let $\chi_i$ be (arbitrary) cutoff functions on $\sph^3$ which are identically $1$ near $p_i$ and $0$ near $-p_i$. Set $\rm{Err}_{\deld, p_i, Q_i} = \tau^{-3} \deld(\chi_i g_{p_i,\mf{m}_i, Q_i} + (1-\chi_i) g_{dS}, \chi_i A_{Q_i})(\tau\del_\tau)\vert_{\tau = 0}$ and $\rm{Err}_{\deld} = \sum_{i=1}^N \rm{Err}_{\deld,p_i, Q_i}$. Then $\rm{Err}_{\deld}$ is a normal $1$-form and the naively glued metric and potential $g_{(3)},A_{(2)}$ from \cref{eq:DefinitionMetricg_(3),eq:DefinitionPotentialA_(2)}
    satisfy \[
    \deld(g_{(3)},A_{(2)}) = \tau^3 \rm{Err}_{\deld} \mod \tau^4 \mc{C}^\infty.\]
    We have $\int_{\sph^3} (\rm{Err}_{\deld})_N \rm{d}g_{\sph^3} = 0$ if and only if the charge balance condition
    $\sum_{i=1}^N Q_i = 0$
    is satisfied.
\end{lemma}
\begin{proof}
We have just shown this for the case $\chi = \chi(t)$ and a single black hole. For multiple black holes, it follows from linearity of the indicial families. \\
If now the cutoff functions $\chi_i$ are arbitrary, we may again choose additional cutoff functions $\tilde{\chi}_i = \tilde{\chi}_i(t)$ and write $\chi_i = \tilde{\chi}_i + (\chi_i - \tilde{\chi}_i)$. As before, the $\mc{O}(\tau_s^3)$ error produced by the $(\chi_i - \tilde{\chi}_i)A_{p_i,Q_i}$ terms is now $(-\delta_{h_s} (\chi_i - \tilde{\chi}_i) (\alpha_2)_T,0)$, where $\alpha_2$ is the order $\tau_s^2$ part of $A_{p_i,Q_i}$. But because $\chi_i - \tilde{\chi}_i$ is $0$ near $\pm p_i$ integrating this over $\del M$ will give $0$ by Stokes' theorem.
\end{proof}
Let us now show that the charge balance condition is also sufficient for the correction term $\tilde{\alpha}$ to exist.
\begin{proof}[Proof of the $\deld$ part of \cref{prop:CorrectionOfObstructionError}]
    We do this using cohomology with compact support. Because $\chi_i\equiv 1$ on a neighborhood $\overline{V_{p_i}}$ of $p_i$, $\rm{Err}_{\deld}$ vanishes on $\overline{V_{p_i}}$. We may therefore choose an open \emph{connected} $\Omega \subset \sph^3$ with $\mathrm{supp}\  \mathrm{Err}_{\deld} \subset \Omega$ and $\Omega \cap \bigcup_{i=1}^N \overline{V_{p_i}} = \emptyset$. Then standard theory in de Rham cohomology (see e.g. \cite[Theorem 17.30]{LeeManifolds}) implies that there exists a $2$-form $\omega = \star \tilde{\alpha}$ supported in $\Omega$ with $\rm{d}\omega = - (\rm{Err}_{\deld})_N \rm{d}g_{\sph^3}$ if and only if
\begin{align*}
0 &= \int_{\Omega} (\rm{Err}_{\deld})_N \rm{d}g_{\sph^3} \\
&= \int_{\sph^3} (\rm{Err}_{\deld})_N \rm{d}g_{\sph^3}
\end{align*}
which is exactly the condition from before. \\
We consequently let $A_{(3)} \coloneqq A_{(2)} + \tau^2\tilde{\alpha}$ and show that this indeed cancels the error: 
\begin{align*}
    \deld(g_{(3)},A_{(3)}) &= \deld(g_{(3)},A_{(2)}) + \deld(g_{(3)}, \tau^2\tilde{\alpha}) \\
    &\equiv \rm{Err}_{\deld} + I(K_{0,g_{(3)},0}[\tau], 2)(0,\tau^2\tilde{\alpha}) \mod \tau^4\mc{C}^\infty \\
    &\equiv \rm{Err}_{\deld} + I(K_{0,g_{dS},0}[\tau], 2)(0,\tau^2\tilde{\alpha}) \mod \tau^4\mc{C}^\infty \\
    &= 0;
\end{align*}
here we have used \cref{lemma:StabilityOfLinearization} and \cref{lemma:DifferentialOperatorsUpToLinearization} in lines two and three respectively.
\end{proof}
We now turn towards proving the Einstein part of \cref{prop:CorrectionOfObstructionError}. We show that $P(g_{(3)}, A_{(3)}) = \mc{O}(\tau^4)$ and that this order $4$ error may be corrected using a $\mc{O}(\tau^3)$ correction to $g_{(3)}$ supported away from the $p_i$. A small issue we face is that because $A_{(3)} = \mc{O}(\tau^2)$, we may not calculate the order $\tau_s^4$ error using the linearization of $P$ at $A=0$, as \cref{lemma:DifferentialOperatorsUpToLinearization} then only provides agreeance up to order $\tau_s^4$. \\
Instead, we will split $P = P_0 - 4T$, where $P_0(g) = 2(\Ric(g) - \Lambda g)$ as before, use the linearization of $P_0$ and work with $T$ directly. Recall also our notation $L_{0,g}$ for the linearization of $P_0$ at $g$. \\
Let us pretend again that we want to glue a single black RNdS black hole into de Sitter space at $p_0=(1,0,0,0)$, with a radial cutoff function $\chi = \chi(t)$ as in \eqref{eq:singleblackholegluing}. Of course, we cannot hope (and are not trying) to correct $g_{(3)}$ to a solution of $P_0$ with terms supported away from $p_0$, but knowing the error to $P_0$ will help us in calculating the error to $P$ later.
\begin{lemma}\label{lemma:CalculationErrP0s}
    We have
    \[
    P_0(g_{(3)}) \equiv \tau_s^4 \rm{Err}_{P_0,s} \mod \tau_s^5 \mc{C}^\infty
    \]
    where $Err_{P_0,s} = 2 \frac{\rm{d}\tau_s}{\tau_s} \otimes_s \frac{12\mf{m}}{\Lambda}\frac{\rm{d}\chi}{\tau_s} + 4 \chi T(g_{dS}, \alpha)_4$. Here, a subscript 4 denotes the order $\tau_s^4$ part of a tensor.
\end{lemma}
\begin{proof}
    Because the de Sitter metric solves the Einstein vacuum equations
    \begin{align*}
        P(g_{(3)}, A_{(2)}) &= P_0(g_{(3)}) - 4 T(g_{(3)}, A_{(2)}) \\
        &\equiv L_{0,g_{dS}}(\chi \gamma) - 4 T(g_{dS},A_{(2)}) \mod \tau_s^6\mc{C}^\infty.
    \end{align*}
    Obtaining the order of congruence uses \cref{lemma:DifferentialOperatorsUpToLinearization} again, this time for $P_0(g) = P(g,0)$. We have also used that $T(g_{(3)}, A_{(2)}) - T(g_{dS},A_{(2)}) \in \tau_s^7 \mc{C}^\infty$, which is a consequence of the fact that $T$ is quadratic in its $A$ argument, $A_{(2)} \in \tau_s^2\mc{C}^\infty$ and $\gamma \in \tau_s^3\mc{C}^\infty$.\\
    For $\chi \equiv 1$ the right side will vanish, this time because the RNdS metric and potential solve the Einstein equations. Hence, collecting terms of different orders:
    \begin{gather*}
        I(L_{0,g_{dS}},3)(\gamma_3) = 0, \\
        I(L_{0,g_{dS}}[\tau_s],3)(\gamma_3) + I(L_{0,g_{dS}},4)(\gamma_4) - 4 T(g_{dS},\alpha)_4 = 0.
    \end{gather*}
    It follows that the order $\tau_s^3$ error to $P_0$ is
    $I(L_{0,g_{dS}},3)(\chi \gamma_3) = \chi I(L_{0,g_{dS}},3)(\gamma_3) = 0$.
    The order $\tau_s^4$ error is
    \begin{align*}
        I(L_{0,g_{dS}}[\tau_s],3)(\chi\gamma_3) &+ I(L_{0,g_{dS}},4)(\chi\gamma_4)\\
        &= I(L_{0,g_{dS}}[\tau_s],3)(\chi\gamma_3) + \chi I(L_{0,g_{dS}},4)(\gamma_4)\\
        &=I(L_{0,g_{dS}}[\tau_s],3)(\chi\gamma_3) - \chi I(L_{0,g_{dS}}[\tau_s],3)(\gamma_3) + 4 \chi T(g_{dS},\alpha)_4 \\
        &=[I(L_{0,g_{dS}}[\tau_s],3),\chi](\gamma_3) + 4 \chi T(g_{dS},\alpha)_4 \\
        &= \frac{\Lambda}{3}(0,[-3\delta_{h_s},\chi](\gamma_3)_{TT0},0,0) + 4 \chi T(g_{dS},\alpha)_4\\
        &= (0,\frac{12\mf{m}}{\Lambda}\rm{d}\chi,0,0)+ 4 \chi T(g_{dS},\alpha)_4
    \end{align*}
    where we have used the same calculation of ${[\delta_{h_s},\chi}]$ as in the end of \cref{lemma:ErrorDeltaD}.
\end{proof}
Continue to pretend that we want to glue a single RNdS black hole and that we have added some (for now arbitrary) correction term $\tilde{A}=\mc{O}(\tau_s^2)$ supported away from $p_0$ to $A_{(2)}$.
Then as $A_{(2)} + \tilde{A} = \mc{O}(\tau_s^2)$, $T(g_{(3)}, A_{(2)} + \tilde{A})$ will be $\mc{O}(\tau_s^4)$. \par
The order $\tau_s^3$ error to $P(g_{(3)},A_{(2)} + \tilde{A}) = P_0(g_{(3)}) -4T(g_{(3)}, A_{(2)} + \tilde{A})$ will therefore vanish; the order $\tau_s^4$ error is 
\[\rm{Err}_{P,s} := 2 \frac{\rm{d}\tau_s}{\tau_s} \otimes_s \frac{12\mf{m}}{\Lambda}\frac{\rm{d}\chi}{\tau_s} + 4\chi T(g_{dS},\alpha)_4 - 4T(g_{dS}, A_{(2)} + \tilde{A})_4\]
We will now illustrate how to go about finding a correction to this error.
    Let us start by trying to find an order $\tau_s^4$ correction to $g_{(3)}$ which cancels with the two terms involving $T$. We treat the $-4\chi T(g_{dS},\alpha)_4 =: R$ term; the other one may be treated in the same manner. \par
    For such a correction to $g_{(3)}$ to exist, $R$ needs to be in the image of $I(L_{0,g_{dS},0},4)$, which according to \eqref{eq:linearizationP} is $\R(3 \tau_s^{-2}\rm{d}\tau_s^2 + h_s) + \tau_s^{-2} \rm{ker} \Tr_{h_s}$. We hence need to show that the normal-tangential part of $R$ vanishes and that $R_{NN} - 3 R_{TT1} = 0$. \par
    We have $T(g_{dS},\alpha)_{\mu\nu} = \rm{d}\alpha\indices{_\mu^\rho}\rm{d}\alpha\indices{_\nu_\rho} - \rm{d}\alpha\indices{^\rho^\lambda}\rm{d}\alpha\indices{_\rho_\lambda}(g_{dS})_{\mu\nu}$. Now, as $\alpha = \mc{O}(\tau_s^2)$ and $T$ is square in the $\rm{d}\alpha$ terms, the order $\tau_s^4$ terms of $T$ are made up only of the $\tau_s^2$ terms of $\rm{d}\alpha$. A look at the indicial family of $\rm{d}$, see \cref{lemma:splittingdeltagd}, shows that these must be normal-tangential. It follows from this and a short calculation that $(T(g_{dS},\alpha)_4)_{NT} = 0$ and $(T(g_{dS},\alpha)_4)_{NN} - 3(T(g_{dS},\alpha)_4)_{TT1} = 0$. Hence, $R$ has the required properties as well. Observe that, because we required $\tilde{A}$ to be supported away from $p_0$, the error term vanishes near $p_0$. Therefore, our correction to $g_{(3)}$ we found to cancel the $T$ terms will vanish near $p_0$ as well, so we do \emph{not} change the metric near the gluing point. \par
    The $(0,\frac{12\mf{m}}{\Lambda} \rm{d}\chi,0,0) =: S$ error is just the error for the gluing of a single Schwarzschild-de Sitter black hole. As $T$ does not contribute to the linearization of $P$, this error term may be treated in the same way as in the Schwarzschild-de Sitter case with the Einstein vacuum equations. We thus refer to \cite[\S 3.1.]{blackholegluing} for the detailed calculations and only give a summary here.\par
    Because $S$ is normal-tangential and $I(L_{0,g_{dS},0},4)_{NT}=0$, it cannot be corrected by an order $\tau_s^4$ term; we will instead try to correct it by a term one order lower. We are hence looking for a symmetric $2$-tensor vanishing near $p_0$ (a correction to $g_{(3)}$), which is in the kernel of $I(L_{0,g_{dS},0},3)$ and whose image under $I(L_{0,g_{dS},0}[\tau_s],3)$ is exactly the negative of $S$. An inspection of the form of these indicial families shows that this means finding a symmetric $2$-tensor $k \in \mc{C}^\infty(\del M; \rm{ker} \Tr_{h_s})$ which solves the underdetermined elliptic equation $-(\Lambda/3) 3 \delta_{h_s} k = -12\mf{m}\Lambda^{-1} \rm{d}\chi$. A necessary condition is that $-12\mf{m}/\Lambda \rm{d}\chi$ is $L^2(\del M, \abs{\rm{d}h_s})$ orthogonal to the kernel of formal adjoint of $\delta_{h_s}$. \par
    This kernel consists of the conformal Killing 1-forms of $\del M$, which are in one-to-one correspondence to the conformal Killing vector fields of $\del M$. The condition becomes
    \begin{equation}\label{eq:masscondition}
    \int_{\del M}V\Big(\frac{12 \mf{m}}{\Lambda} \rm{d}\chi\Big) |\rm{d}h_s| = 0
    \end{equation}
    for all conformal Killing vector fields $V$ of $\del M$. Changing to the $\tau$ coordinates \eqref{deSitter:TimeTimesSphere} and using conformal invariance, one may directly calculate what this means using the conformal Killing vector fields of the $3$-sphere. Namely, the condition becomes
    \begin{equation}\label{eq:MassBalanceCondtionPrelim}
    \left<\mf{m}p_0,q\right>_{\R^4} = 0 \ \forall \ q \in \sph^3.
    \end{equation}
    As this is impossible (unless we are in the trivial case $\mf{m} = 0$), we again get the result that gluing a single black hole into de Sitter space in this fashion is impossible. For multiple black holes, we are on the other hand led to the mass balance condition:
\begin{lemma}
        Consider the naively glued metric and potential $g_{(3)},A_{(2)}$, along with the corrected potential $A_{(3)}$ from \cref{prop:CorrectionOfObstructionError}. Here the $\chi_i$ are (arbitrary) cutoff functions on $\sph^3$ identically $1$ near $p_i$ and $0$ near $-p_i$. Set $\rm{Err}_{P_0,p_i, \mf{m}_i,Q_i} := \tau^{-4}P_0(\chi_i g_{p_i,\mf{m}_i, Q_i} + (1-\chi_i)g_{dS})(\tau \cdot, \tau \cdot)\vert_{\tau = 0}$ and $\rm{Err}_{P_0} = \sum_{i=1}^N \rm{Err}_{P_0,p_i, \mf{m}_i,Q_i}$. Then
        \[
        P(g_{(3)}, A_{(3)}) \equiv \tau^4 (\rm{Err}_{P_0} -4T(g_{(3)}, A_{(3)})_4) \mod \tau^5\mc{C}^\infty.
        \]
        Moreover we have
        \[
        \int_{\sph^3} V((\rm{Err}_{P_0})_{NT}) |\rm{d}g_{\sph^3}| = 0
        \]
        if and only the $(p_i, \mf{m_i})$ satisfy the mass balance condition \eqref{eq:MassBalanceCondition}.
    \end{lemma}
    \begin{proof}
        $\rm{Err}_{P_0}$ consists of two parts. The first part consists of the terms stemming from the $T$ terms of the $\rm{Err}_{P_0,p_i,\mf{m}_i,Q_i}$ together with the $-4T(g_{(3)},A_{(3)})_4$ term. As discussed before this proof, these do not have any normal-tangential terms. The second part, the normal-tangential part of $\rm{Err}_{P_0}$, stems from the Schwarzschild-de Sitter part of the $\rm{Err}_{P_0,p_i,\mf{m}_i,Q_i}$. Similar to the proof of \cref{lemma:chargebalancecondition} we may reduce to the case where every $\chi_i$ is a radial cutoff function centered at $p_i$; then \cref{eq:MassBalanceCondtionPrelim} will be
        \[
    \left<\sum_{i=1}^N\mf{m}_ip_i,q\right>_{\R^4} = 0 \ \forall \ q \in \sph^3.
        \]
        instead, which is possible if and only if $\sum_{i=1}^N\mf{m}_ip_i = 0$.
    \end{proof}
    It remains to show that the mass balance condition is not only necessary, but also sufficient for a correction to exist. Again, this was already done in \cite[Proof of Proposition 3.5]{blackholegluing}. We repeat the necessary theorem, which is due to Delay, here.
    \begin{theorem}[{\cite[Theorem 3.10]{blackholegluing}, \cite{Delay2012}}]
        Let (X,h) be a smooth Riemannian manifold and let $\Omega \subset X$ be open. Let $f \in \mc{C}^\infty(X;T^* X)$ satisfy $\rm{supp}(f) \Subset \Omega$ and $\int_\Omega V(f) |\rm{d}h| = 0$ for all conformal Killing vector fields $V$ of $(\Omega,h)$. Then there exists a traceless $k \in \mc{C}^\infty(X; S^2 T^*X)$ with $\rm{supp}\ k \subset \overline{\Omega}$ with $\delta_h k = f$.
    \end{theorem}
    \begin{proof}[Proof of the $P$ part of \cref{prop:CorrectionOfObstructionError}]
        The $\chi_i$ and hence $(\rm{Err}_P)_{NT}$ vanish on a neighborhood $\overline{V_{p_i}}$. As in the proof of the $\deld$ part, we may therefore choose an open \emph{connected} $\Omega \subset \sph^3$ with $\overline{\Omega}$ \emph{disjoint} from the $\overline{V_{p_i}}$ such that $\rm{supp}\ (\rm{Err}_P)_{NT} \Subset \Omega$. Because we chose $\Omega$ to be connected, the conformal Killing vector fields of $\Omega$ and $\sph^3$ coincide. This follows from the fact that the maximum number of independent conformal Killing vector fields on any connected $n$-dimensional Riemannian manifold is at most $(n+1)(n+2)/2$, and $\sph^n$ attains this number.
        Thus
        \begin{align*}
        &\int_\Omega V((\rm{Err}_P)_{NT}) |\rm{d}g_{\sph^3}|= 0 \ \forall V \ \text{conformal Killing on } \Omega \\
        \Leftrightarrow &\int_{\sph^3} V((\rm{Err}_P)_{NT}) |\rm{d}g_{\sph^3}| = 0 \ \forall V \ \text{conformal Killing on } \sph^3
        \end{align*}
        which is just the mass balance condition. Together with the discussion after \cref{lemma:CalculationErrP0s}, we thus find symmetric $2$-tensors $k, k' \in \mc{C}^\infty(\del M, S^2 T^*M\vert_{\del M})$, supported away from the $p_i$, such that \begin{gather*}
            I(L_{0,g_{dS},0},3)k = 0 \\
            I(L_{0,g_{dS},0}[\tau],3)k + I(L_{0,g_{dS},0},4)k' = -(\rm{Err}_{P_0}-4T(g_{(3)},A_{(3)})
        \end{gather*} We therefore define $g_{(4)} = g_{(3)} + \tau^3 k + \tau^4 k'$ and show that this indeed cancels out the error to $P(g_{(3)},A_{(3)})$:
        \begin{align*}
            P(g_{(4)},A_{(3)}) &\equiv P(g_{(3)}, A_{(3)}) + L_{0, g_{(3)}, A_{(3)}}(\tau^3 k + \tau^4 k',0) \mod \tau^6 \mc{C}^\infty \\
            &\equiv  P(g_{(3)}, A_{(3)}) + L_{0, g_{dS}, 0}(\tau^3 k + \tau^4 k',0) \mod \tau^5 \mc{C}^\infty \\
            &= 0
        \end{align*}
        Here we have used \cref{lemma:DifferentialOperatorsUpToLinearization} and \cref{lemma:StabilityOfLinearization}.
        By the same lemmas the order $4$ error to $\deld$ is also not changed if we make the change from $g_{(3)}$ to $g_{(4)}$:
        \begin{align*}
            \deld(g_{(4)},A_{(3)}) &\equiv \deld(g_{(3)}, A_{(3)}) + L_{0, g_{(3)}, A_{(3)}}(g_{(4)} - g_{(3)},0) \mod \tau^6 \mc{C}^\infty \\
            &\equiv  \deld(g_{(3)}, A_{(3)}) + L_{0, g_{dS}, 0}(g_{(4)} - g_{(3)},0) \mod \tau^5 \mc{C}^\infty \\
            &= 0 \qedhere
        \end{align*}
    \end{proof}
    \begin{remark}
        We could have also applied Delay's results to find the correction to $A_{(2)}$. According to \cite[\S 9.1]{Delay2012}, a necessary and sufficient condition for a solution $\tilde{\alpha} \in \mc{C}^\infty(\del M;T^*\del M)$ to $-\delta_{h_s}\tilde{\alpha} = -(\rm{Err}_{\deld, s})_N$ with support contained in the $\Omega$ from the proof of the $\deld$ part to exist, $(\rm{Err}_{\deld, s})_N$ has to be $L^2(\Omega, |\rm{d}h_s|)$ orthogonal to the kernel of the formal adjoint of $\delta_{h_s}$. This formal adjoint is the exterior derivative $d$, so the kernel consists of the locally constant functions. As $\Omega$ was chosen connected, these are just the constant functions, so we need
        \[
        0 = \int_{\Omega}(\rm{Err}_{\deld, s})_N |\rm{d}h_s|
        \]
        which is just the same condition as before.
    \end{remark}

%% file: toinfinityandbeyond.tex
\subsection{Construction of the formal solution}
We now find additional perturbations to the metric and potential constructed in \cref{prop:CorrectionOfObstructionError} to give a formal solution of the Einstein--Maxwell equations. Our strategy will be to alternate between adding a correction to the potential to increase the order of $\deld$ and adding a correction to the metric to increase the order of $P$. We will construct these corrections using an exactness argument. Namely, we start with the general fact that $\delta_g^2 = 0$ on differential forms for any Lorentzian metric $g$; in our case $g = g_{dS}$. Considering only leading order terms of $0 = \delta_{g_{dS}} \circ \delta_{g_{dS}}\rm{d}$, and using the fact that $I(\delta_{g_{dS}}\rm{d},\lambda)$ and $I(K_{0,g_{dS},0},\lambda)$ are the same maps when acting on covectors, we obtain the sequence 
\begin{equation}\label{eq:sequenceK_0}
    \begin{tikzcd}[column sep=large]
    	{\mc{C}^\infty(\del M; \nT^*M\vert_{\del M})} & {\mc{C}^\infty(\del M; \nT^*M\vert_{\del M})} & {\mc{C}^\infty(\del M).}
    	\arrow["{I(K_{0,g_{dS},0},\lambda)}", from=1-1, to=1-2]
    	\arrow["{I(\delta_{g_{dS}},\lambda)}", from=1-2, to=1-3]
    \end{tikzcd}
\end{equation}
We find a similar sequence for $P$. Indeed, the second Bianchi identity implies $0 = \delta_g G_g P_0(g)$ for any Lorentzian metric $g$. Therefore, if we let $g_s = g_{dS} + s\tilde{g}$ for some $\tilde{g} \in \mc{C}^\infty(M; S^2 \nT^*M)$, then because the de Sitter metric solves $P_0(g_{dS})=0$, differentiating yields
\begin{align*}
    0 &= \frac{d}{ds}\Big\vert_{s=0} \delta_{g_s}G_{g_s}P_0(g_s)\\
        &= \delta_{g_{dS}}G_{g_{dS}} L_{0,g_{dS}}(\tilde{g}).
\end{align*}
Using that the linearization of $P_0$ at $g=g_{dS}$ and $P$ at $g=g_{dS}, A=0$ acting on symmetric 2-tensors are the same maps, we again obtain a sequence
\begin{equation}\label{eq:sequenceL_0}
    \begin{tikzcd}[column sep=large]
    	{\mc{C}^\infty(\del M; S^2 \nT^*M\vert_{\del M})} & {\mc{C}^\infty(\del M; S^2 \nT^*M\vert_{\del M})} & {\mc{C}^\infty(\del M; \nT^*M\vert_{\del M}).}
    	\arrow["{I(L_{0,g_{dS},0},\lambda)}", from=1-1, to=1-2]
    	\arrow["{I(\delta_{g_{dS}}G_{g_{dS}},\lambda)}", from=1-2, to=1-3]
    \end{tikzcd}
\end{equation}
\begin{lemma}\label{lemma:SequencesAreExact}
The sequence \eqref{eq:sequenceK_0} is exact for $\lambda \geq 4$ and \eqref{eq:sequenceL_0} is exact for $\lambda \geq 5$. Moreover, the resulting preimages for the sequences' first maps, ${I(K_{0,g_{dS},0},\lambda)}$ and $ {I(L_{0,g_{dS},0},\lambda)}$, may be chosen to have the same support as their images.
\end{lemma}
\begin{proof}
    This is straightforward to check. According to \cref{lemma:splittingsDiffOpOnSymmetric} together with a short manual calculation for $I(\delta_g,\lambda)$ we get that in the splitting \eqref{eq:refinedsplitting}
    \begin{gather*}
        \frac{3}{\Lambda}I(\delta_gG_g, \lambda) = \begin{pmatrix}
            \frac{1}{2}(\lambda - 6) & 0 & \frac{3}{2}(\lambda - 2) & 0\\
            0 & \lambda - 4 & 0 & 0\\
        \end{pmatrix}
        \\
        \frac{3}{\Lambda}I(\delta_g, \lambda) = \begin{pmatrix}
            \lambda - 3 & 0
        \end{pmatrix}
    \end{gather*}
    Therefore, using the matrix representations of $I(L_{0,g_{dS},0}, \lambda)$ and $I(K_{0,g_{dS},0}, \lambda)$ obtained in \eqref{eq:linearizationP}, \eqref{eq:linearizationdeltad}, we see for $\lambda \geq 5$:
    \begin{align*}
        \Ker I(\delta_gG_g, \lambda) &= \R\tau^{-2}\Big(\rm{d}\tau_s^2 + \frac{\lambda-6}{3(\lambda-2)}h_s\Big) + \tau^{-2}\Ker \Tr_{h_s} \\ 
        &= \rm{Im} I(L_{0,g_{dS},0},\lambda)\vert_{\mc{C}^\infty(\del M; S^2 \nT^*M\vert_{\del M})},
    \end{align*}
    and for $\lambda \geq 4$:
    \[
    \Ker I(\delta_g, \lambda) = \tau^{-1}T^*X = \rm{Im} I(K_{0,g_{dS},0},\lambda)\vert_{\mc{C}^\infty(\del M; \nT^*M\vert_{\del M})}.
    \]
\end{proof}
Because the sequences are exact, we are handed a way of proving the existence of error corrections. Namely, if we show that the current error to $\deld$, respectively $P$, is in the kernel of $I(\delta_{g_{dS}}, \lambda), \lambda \geq 4$, respectively $I(\delta_{g_{dS}}G_{g_{dS}},\lambda),\lambda \geq 5$, then by exactness these error terms (and their negatives) will be in the image of the indicial families of the respective linearization. We may thus choose an element in the preimage, supported on the same set as the error, as a candidate for the correction. To make this rigorous, we need a simple lemma ensuring that the electromagnetic tensor will not interfere.
\begin{lemma}\label{lemma:OrderOfdeltaT}
    Let $M$ be a $(n+1)$ dimensional manifold with boundary defining function $\tau$ and Lorentzian $0$-metric $g$. Suppose $A \in \tau^2\mc{C}^\infty(M; \nT^*M)$ is a vector potential with $\deld(g, A) \in \tau^k\mc{C}^\infty(M; \nT^*M)$. Then \[
    \delta_g G_g T(g,A) \in \tau^{k+2}\mc{C}^\infty(M; \nT^*M). 
    \]
\end{lemma}
\begin{proof}
    This is a direct calculation. First, observe that $T$ is traceless, so $G_g T = T$. Therefore,
    \begin{equation}\label{eq:calculationT}
    \begin{split}
        \delta_g G_g T(g,A)_\nu &= \delta_g T(g,A)_\nu \\ &= -\nabla^\mu \rm{d}A\indices{_\mu^\alpha}\rm{d}A\indices{_\nu_\alpha} - \rm{d}A\indices{_\mu^\alpha}\nabla^\mu \rm{d}A\indices{_\nu_\alpha} + \frac{1}{2}  \rm{d}A^{\alpha\beta}\nabla_\nu \rm{d}A_{\alpha\beta}
    \end{split}
    \end{equation}
    Written out in coordinates the vanishing of the second exterior derivative means that $\nabla_\nu \rm{d}A_{\alpha\beta} + \nabla_\alpha \rm{d}A_{\beta\nu} + \nabla_\beta \rm{d}A_{\nu\alpha} = 0$. Using this on the third term above gives
    \begin{align*}
         \rm{d}A^{\alpha\beta} \nabla_\nu \rm{d}A_{\alpha\beta} &= -  \rm{d}A^{\alpha\beta} \nabla_\alpha \rm{d}A_{\beta\nu} - \rm{d}A^{\alpha\beta}\nabla_\beta \rm{d}A_{\nu\alpha} \\
        &= - 2 \rm{d}A^{\alpha\beta} \nabla_\beta \rm{d}A_{\nu\alpha}
    \end{align*}
    where we have used antisymmetry of $\rm{d}A$ in the second line. The second and third term in \eqref{eq:calculationT} will hence cancel. As the first term is just $-\delta_g \rm{d}A^\alpha \rm{d}A_{\nu\alpha}$ the claim now follows from the prescribed orders of vanishing for $A$ and $\delta_g \rm{d}A$.
\end{proof}

The next three propositions illustrate the procedure for increasing the order of the error term in detail. Recall the already constructed $g_{(3)},g_{(4)}$ and $A_{(3)}$ from \cref{prop:CorrectionOfObstructionError}.

\begin{proposition}\label{prop:correctionDeltaD}
    Let $\lambda \geq 4$ and let $g_{(\lambda-1)}$ be some smooth Lorentzian 0-metric with $g_{(\lambda-1)} - g_{dS} \in \tau^3 \mc{C}^\infty$. Suppose $A_{(\lambda-1)} \in \tau^2\mc{C}^\infty(M; \nT^*M)$ satisfies $\deld(g_{(\lambda - 1)}, A_{(\lambda-1)}) \in \tau^\lambda\mc{C}^\infty$. Then there exists $A_{(\lambda)} \in \tau^2\mc{C}^\infty(M; \nT^*M)$ such that $A_{(\lambda)} - A_{(\lambda -1)} \in \tau^{\lambda}\mc{C}^\infty$ and $\deld(g_{(\lambda-1)}, A_{(\lambda)}) \in \tau^{\lambda+1}\mc{C}^\infty$.
\end{proposition}
\begin{proof}
    Let $f_\lambda \coloneqq \tau^{-\lambda}\deld (g_{(\lambda -1)},A_{(\lambda-1)})\vert_{\tau=0}$ be the order $\tau^\lambda$ error to $\deld$.
    Suppose $A' \in \tau^\lambda \mc{C}^\infty(M; \nT^*M)$ is some arbitrary correction term. Then, by \cref{lemma:DifferentialOperatorsUpToLinearization} and \cref{lemma:StabilityOfLinearization},
    \begin{align*}
    \deld(g_{(\lambda-1)},A_{(\lambda-1)} + A') &\equiv \deld(g_{(\lambda -1)},A_{(\lambda -1)}) + K_{0, g_{(\lambda-1)},A_{(\lambda -1)}}(0, A') \mod \tau^{2\lambda}\mc{C}^\infty \\
    &\equiv \deld(g_{(\lambda -1)},A_{(\lambda -1)}) + K_{0, g_{dS},0}(0, A') \mod \tau^{\lambda + 1}\mc{C}^\infty \\
    \end{align*}
    On the other hand, making use of $\delta \circ \delta = 0$, we obtain
    \begin{align*}
        0 &= \delta \deld(g_{(\lambda-1 )},A_{(\lambda-1)}) \\
        &\equiv \delta_{g_{(\lambda-1)}}\tau^\lambda f_\lambda \mod \tau^{\lambda + 1}\mc{C}^\infty \\
        &\equiv \delta_{g_{dS}}\tau^\lambda f_\lambda \mod \tau^{\lambda + 1}\mc{C}^\infty \\
        &\equiv \tau^\lambda I(\delta_{g_{dS}},\lambda) f_\lambda \mod \tau^{\lambda + 1}\mc{C}^\infty.
    \end{align*}
    In other words $f_\lambda \in \rm{ker}I(\delta_{g_{dS}},\lambda) = \rm{im}I(K_{0,g_{dS},0},\lambda)$.
    We may thus indeed choose $A'$ to exactly cancel out the $\mc{O}(\tau^\lambda)$ error to $\deld$. The proof is finished upon defining $A_{(\lambda)} \coloneqq A_{(\lambda-1)} + A'$.
\end{proof}
\begin{proposition}\label{prop:correctionP}
    Let $\lambda \geq 5$ and $A_{(\lambda-1)} \in \tau^2\mc{C}^\infty(M; \nT^*M)$ be a potential. Suppose $g_{(\lambda-1)}$ is a smooth Lorentzian $0$-metric with $g_{(\lambda-1)} - g_{dS} \in \tau^3 \mc{C}^\infty$. Assume $P(g_{(\lambda-1)}, A_{(\lambda - 1)}) \in \tau^{\lambda}\mc{C}^\infty$ and $\deld(g_{(\lambda -1)},A_{(\lambda-1)}) \in \tau^{\lambda}\mc{C}^\infty$. Then there exists a smooth Lorentzian $0$-metric $g_{(\lambda)}$ such that $g_{(\lambda)} - g_{(\lambda-1)} \in \tau^\lambda\mc{C}^\infty$ and $P(g_{(\lambda)}, A_{(\lambda-1)}) \in \tau^{\lambda +1}\mc{C}^\infty$.
\end{proposition}
\begin{proof}
    Let $s_\lambda = \tau^{-\lambda}P(g_{(\lambda-1)}, A_{(\lambda - 1)})\vert_{\tau =0}$ be the order $\tau^\lambda$ error to $P$.
    Then 
    \begin{align*}
        P_0(g_{(\lambda -1)}) \equiv 4T(g_{(\lambda - 1)}, A_{(\lambda-1)}) + \tau^\lambda s_\lambda \mod \tau^{\lambda + 1}\mc{C}^\infty
    \end{align*}
    Therefore, using $0 = \delta G \circ P_0$ and \cref{lemma:OrderOfdeltaT}, 
    \begin{align*}
        0 &= \delta_{g_{(\lambda-1)}} G_{g_{(\lambda-1)}}P_0(g_{(\lambda-1)}) \\
        &\equiv \delta_{g_{(\lambda-1)}} G_{g_{(\lambda-1)}}(4T(g_{(\lambda - 1)}, A_{(\lambda-1)}) + \tau^\lambda s_\lambda) \mod \tau^{\lambda + 1}\mc{C}^\infty \\
        &\equiv \delta_{g_{(\lambda-1)}} G_{g_{(\lambda-1)}}\tau^\lambda s_\lambda \mod \tau^{\lambda + 1}\mc{C}^\infty \\ 
        &\equiv \delta_{g_{dS}} G_{g_{dS}} \tau^\lambda  s_\lambda\mod \tau^{\lambda + 1}\mc{C}^\infty \\
        &\equiv  \tau^\lambda  I(\delta_{g_{dS}} G_{g_{dS}}, \lambda) s_\lambda\mod \tau^{\lambda + 1}\mc{C}^\infty \\
    \end{align*}
    Thus $s_\lambda \in \ker I(\delta_{g_{dS}} G_{g_{dS}}, \lambda) = \rm{im} I(L_{0,g_{dS},0},\lambda)$ and we may choose some $k \in \mc{C}^\infty(\del M; S^2 \nT^*M)$ to cancel $I(\delta_{g_{dS}} G_{g_{dS}}, \lambda) s_\lambda$. Defining $g_{(\lambda)} \coloneqq g_{(\lambda - 1)} + \tau^\lambda k$ we obtain
    \begin{align*}
        P(g_{(\lambda-1)}, A_{(\lambda)}) &\equiv P(g_{(\lambda-1)}, A_{(\lambda)}) + L_{0,g_{(\lambda-1)}, A_{(\lambda)}}(\tau^\lambda k, 0) \mod \tau^{2\lambda}\mc{C}^\infty \\
        &\equiv P(g_{(\lambda-1)}, A_{(\lambda - 1)}) + L_{0,g_{dS}, 0}(\tau^\lambda k, 0) \mod \tau^{\lambda +1}\mc{C}^\infty\\
        &\equiv 0 \mod \tau^{\lambda +1}\mc{C}^\infty. \qedhere
    \end{align*}
\end{proof}

We may now finally prove:

\begin{proposition}\label{prop:CorrectionToInfiniteOrder}
    Under the assumptions of \cref{theorem:BlackHoleGluing}, let $g_{(4)}$ and $A_{(3)}$ be as in \cref{prop:CorrectionOfObstructionError}. Then there exist a smooth Lorentzian $0$-metric $g_{(\infty)} \in \mc{C}^\infty(M;S^2 \nT^*M)$ and $A_{(\infty)} \in \mc{C}^\infty(M; \nT^*M)$ such that
    \begin{enumerate}
        \item $g_{(\infty)} - g_{(4)} \in \tau^5 \mc{C}^\infty$ and $g_{(\infty)} = g_{(4)}$ near $\bigcup_{i=1}^N \overline{V_{p_i}}$;
        \item $A_{(\infty)} - A_{(3)} \in \tau^4 \mc{C}^\infty$ and $A_{(\infty)} = A_{(3)}$ near $\bigcup_{i=1}^N \overline{V_{p_i}}$;
        \item $P(g_{(\infty)},A_{(\infty)}) \in \tau^\infty \mc{C}^\infty$;
        \item $\deld(g_{(\infty)},A_{(\infty)}) \in \tau^\infty \mc{C}^\infty$.
    \end{enumerate}
    Here $\tau^\infty \mc{C}^\infty = \bigcap_{m=1}^\infty \tau^m \mc{C}^\infty$.
\end{proposition}
\begin{proof}
    We may apply \cref{prop:correctionDeltaD} to find $A_{(4)} \in \tau^2\mc{C}^\infty$ with $A_{(4)} - A_{(3)} \in \tau^4\mc{C}^\infty$ and $\deld(g_{(3)},A_{(4)}) \in \tau^5 \mc{C}^\infty$.
    Then
    \begin{align*}
        \deld(g_{(4)}, A_{(4)}) &\equiv \deld(g_{(3)}, A_{(4)}) + K_{0,g_{(3)}, A_{(4)}}(g_{(4)} - g_{(3)}, 0) \mod \tau^{2\cdot 3}\mc{C}^\infty\\
        &\equiv \deld(g_{(3)}, A_{(4)}) + K_{0,g_{dS}, 0}(g_{(4)} - g_{(3)}, 0) \mod \tau^{3+2}\mc{C}^\infty\\
        &\equiv 0 \mod \tau^5\mc{C}^\infty,
    \end{align*}
    where we have used that $K_{0,g_{dS},0}$ vanishes when acting on symmetric $2$ tensors in the last line. \\
    Similarly 
    \begin{align*}
        P(g_{(4)}, A_{(4)}) &\equiv P(g_{(4)}, A_{(3)}) + L_{0,g_{(4)}, A_{(3)}}(0,A_{(4)} - A_{(3)}) \mod \tau^{2\cdot 4}\mc{C}^\infty\\
        &\equiv P(g_{(4)}, A_{(3)}) + L_{0,g_{dS}, 0}(0,A_{(4)} - A_{(3)}) \mod \tau^{4+2}\mc{C}^\infty\\
        &\equiv 0 \mod \tau^5\mc{C}^\infty ,
    \end{align*}
    where we have used that $L_{0,g_{dS},0}$ vanishes when acting on covectors. \\
    We may therefore apply \cref{prop:correctionP} to produce $g_{(5)}$ with $g_{(5)}-g_{(4)} \in \tau^5 \mc{C}^\infty$ and $P(g_{(5)}, A_{(4)}) \in \tau^6\mc{C}^\infty$. \\
    Then inductively repeating this procedure for higher orders gives us, for $\lambda \geq 4$, $g_{(\lambda)},A_{(\lambda)}$ such that
    \begin{gather*}
        \deld (g_{(\lambda)}, A_{(\lambda)}) \in \tau^{\lambda+1} \mc{C}^\infty, \\
        P (g_{(\lambda)}, A_{(\lambda)}) \in \tau^{\lambda+1} \mc{C}^\infty.
    \end{gather*}
    Borel's lemma produces $g_{(\infty)}$ and $A_{(\infty)}$ with 
    \begin{equation}
        \begin{split}
            g_{(\infty)} \equiv g_{(\lambda)} \mod \tau^{\lambda + 1}, \\
            A_{(\infty)} \equiv A_{(\lambda)} \mod \tau^{\lambda + 1} \\
        \end{split}
    \end{equation}
    for all $\lambda \geq 4.$
    It remains to check that these solve the Einstein--Maxwell equations to infinite order. But
    \begin{align*}
        P(g_{(\infty)}, A_{(\infty)}) &\equiv P(g_{(\lambda)}, A_{(\lambda)}) + L_{0, g_{(\lambda)}, A_{(\lambda)}}(g_{(\infty)} - g_{(\lambda)}, A_{(\infty)}-A_{(\lambda)}) \mod \tau^{2\lambda +2}\mc{C}^\infty \\
        &\equiv 0 \mod \tau^{\lambda+1}\mc{C}^\infty
    \end{align*}
    for all $\lambda \geq 4$, and therefore $P(g_{(\infty)}, A_{(\infty)}) \in \tau^\infty \mc{C}^\infty$. Similarly for $\deld$.
\end{proof}

%% file: solvingthenonlinearequations.tex
\subsection{Solving the nonlinear equations}
So far we have constructed, in \cref{prop:CorrectionOfObstructionError} and \cref{prop:CorrectionToInfiniteOrder}, formal solutions $g_{(\infty)}$, $A_{(\infty)}$ to the Einstein--Maxwell system. There remain $\mc{O}(\tau^\infty)$ errors to $\delta d$ and $P$ which we will now solve away by adding additional $\tilde{g}, \tilde{A} \in \tau^\infty \mc{C}^\infty$ corrections. \\
Following \cite[\S 2.2.]{HintzKerrNewmanDeSitter}, we will work in a DeTurck gauge \cite{DeTurck81,GrahamLee} for the metric and a Lorenz-gauge for the electromagnetic potential. Consider a background Lorentzian metric $t$ and electromagnetic potential $B$ on $M$. For an additional Lorentzian $0$-metric $g$ and potential $A$ we define the gauge $1$-form
\begin{equation}\label{eq:defUpsilon}
    \Upsilon(g) \coloneqq g t^{-1} \delta_g G_g t,
\end{equation}
where $gt^{-1}$ is the endomorphism acting on $T^*M$ via $gt^{-1}\omega_\mu = g_{\mu\nu}(t^{-1})^{\nu\rho}\omega_\rho$, and the gauge function
\begin{equation}\label{eq:defGaugeFunction}
    \Upsilon^M(g;A) \coloneqq \Tr_{g}\delta_t^* A - \Tr_{g}\delta_t^* B.
\end{equation}
We then also define the gauge-fixed version of the Einstein--Maxwell system
\begin{equation}\label{eq:defGaugeFixedEinstein-Maxwell}
\begin{split}
    P_{DT}(g;A) &\coloneqq 2(\Ric(g) - \Lambda g - 2T(g, A)- \delta_g^* \Upsilon(g)), \\
    P_L(g;A) &\coloneqq \delta_g \rm{d}A - \rm{d}\Upsilon^M(g;A).
\end{split}
\end{equation}
We view this as a coupled system of differential equations for the correction terms $g-t$ and $A-B$. The $\Upsilon$ terms transform it into a system of quasilinear wave equations in that they cancel second-order terms so that the principal symbol of $(P_{DT}, P_L)$ as an operator acting on $g-t, A-B$ is exactly given by the inverse metric function. In our case, we work with the background metric $t = g_{(\infty)}$ and the background potential $B = A_{(\infty)}$. We will show that the gauge-fixed system may be solved using correction terms $g-g_{(\infty)}$, $A-A_{(\infty)}$ vanishing to infinite order:
\begin{proposition}\label{prop:solutionToGaugeFixedEq}
    Let $g_{(\infty)},A_{(\infty)}$ be the formal solutions from \cref{prop:CorrectionToInfiniteOrder}. Then there exist $\tilde{g} \in \tau^\infty\mc{C}^\infty(M; S^2 \nT^*M)$ and $\tilde{A} \in \tau^\infty \mc{C}^\infty(M; \nT^*M)$, vanishing near $\bigcup_{i=1}^N\overline{V_{p_i}}$, such that $g \coloneqq g_{(\infty)} + \tilde{g}$ and $A \coloneqq A_{(\infty)} + \tilde{A}$ satisfy
    \begin{equation}
        P_{DT}(g;A) = 0, \quad P_L(g;A) = 0 \ \text{near} \ \tau=0.
    \end{equation}
\end{proposition}
Using the identity $\delta_g^2=0$ on $P_L$ will then show that $\Upsilon^M(g;A)$ itself solves a wave equation with vanishing initial data. Uniqueness then implies $\Upsilon^M(g;A) = 0$, so the constructed solution will solve Maxwell's equations. Applying the second Bianchi identity to $P_{DT}$ will similarly show that $\Upsilon(g)$ solves a wave equation and vanishes, i.e. the inhomogenous Einstein equations are satisfied as well.
In the proof of the above proposition, we make use of results on quasilinear wave equations in asymptotically de Sitter spaces from \cite{asy-ds}. We state these now; see also \cite{vasy-waveeq, Zworski2016Dec}. We denote by $\underline{\R}^k \coloneqq M \times \R^k$ the trivial vector bundle on a manifold $M$.
\begin{proposition}\label{prop:UniquenessAndExistenceQuasilinearWave}
    Let $M=[0,1)_\tau \times X$, where $X$ is an $n$-dimensional smooth Riemannian manifold with metric $h$. Suppose $g_{\text{bg}}$ is a smooth Lorentzian $0$-metric, which in local coordinates $\tau\geq 0, x\in \R^n$ is of the form
    \begin{equation}\label{eq:backgroundmetricinasy-ds}
    g_{\text{bg}}(\tau,x,\rm{d}\tau,\rm{d}x) = \frac{3}{\Lambda}\frac{-\rm{d}\tau^2 + h(x,\rm{d}x)}{\tau^2} + \overline{g}(\tau,x,\rm{d}\tau,\rm{d}x),
    \end{equation}
    where $\overline{g} \in \tau^\eta \mc{C}^\infty(M; S^2 \nT^*M)$ for some $\eta > 0$.  Let $P$ be a quasilinear wave operator acting on sections of a smooth vector bundle $B$, such that for any local trivialization $\tau \geq 0, x \in \R^n, u = (u_1, \dots, u_k) \in \R^k$ for $B$ there exist smooth nonlinear bundle maps $G: \underline{\R}^k \rightarrow S^2 \nT^*M$ and $P_1: \underline{\R}^k \times \underline{\R}^{(n+1)k}\rightarrow\underline{\R}^k$ with $G(\tau,x,0) = g_{\text{bg}}\vert_{(\tau,x)}$ and
    \[
    P(u) \coloneqq \Box_{G(\tau,x,u)}u + P_1(\tau,x,u, {}^0\nabla u).
    \]
    Here ${}^0\nabla = (\tau\del_\tau, \tau\del_{x_1}, \dots, \tau\del_{x_n})$ is the $0$-gradient. Lastly, assume $P(0) \in \tau^\infty \mc{C}^\infty(M; B)$. Then there exists a neighborhood $\mc{U}$ of $\{0\}\times X$ such that there exists a unique solution $u \in \mc{C}^\infty(\mc{U}; B)$ to $P(u)=0$.
\end{proposition}
We will proceed to show that the differential operator $(P_{DT}, P_L)$, acting on $(\tilde{g},\tilde{A}) \in S^2 \nT^*M \oplus \nT^*M$, fits the assumptions of \cref{prop:UniquenessAndExistenceQuasilinearWave}, with $G(\tau,x,\tilde{g},\tilde{A}) = g\vert_{(\tau,x)} = g_{(\infty)} + \Tilde{g}\vert_{(\tau,x)}$. We take $g_{\text{bg}} \coloneqq g_{(\infty)}$ as the background metric. Observe that indeed $g_{\text{bg}} - g_{dS} \in \tau^3\mc{C}^\infty(M; S^2 \nT^*M)$, so it is of the form required in \eqref{eq:backgroundmetricinasy-ds}.

Starting with $P_L$, by choice of $\Upsilon^M(g;A)$, $P_L$ does \emph{not} involve any second derivatives of $\tilde{g}$. 
Now,
\begin{equation}\label{eq:expansiondelta_gdA_nu}
\begin{split}
\delta_g \rm{d}A_\nu = - g^{\mu\alpha}(\del_\alpha\del_\mu A_\nu - \del_\alpha \del_\nu A_\mu &- \Gamma(g)\indices{^\beta_{\alpha\mu}} \del_\beta A_\nu  
+\Gamma(g)\indices{^\beta_{\alpha\mu}} \del_\nu A_\beta \\&- \Gamma(g)\indices{^\beta_{\alpha\nu}}  \del_\mu A_\beta + \Gamma(g)\indices{^\beta_{\alpha\nu}} \del_\beta A_\mu).
\end{split}
\end{equation}
Any of the terms in $\tilde{A}$ of order lower than $2$ may be put into $P_1$, by the same procedure as in \cite[\S 4.1]{asy-ds}. Namely, let us denote with an overline coefficients in the $0$-frame $\tau\del_\tau,\tau\del_{x_i}$ and $0$-coframe $\rm{d}\tau/\tau, \rm{d}_{x_i}/\tau$, so for example $\tau^2 g^{\overline{\mu}\overline{\nu}} = g^{\mu\nu}$. Then one of the first order terms appearing in $\delta_g \rm{d}A_{\overline{\nu}}$ is (this one coming from the third term in \eqref{eq:expansiondelta_gdA_nu})
\begin{align*}
    \tau \frac{1}{2}g^{\mu\alpha}g^{\beta\gamma}(\del_\mu g_{\alpha\gamma})(\del_\beta \tilde{A}_\nu) &= \tau \frac12\tau^2g^{\overline{\mu\alpha}}\tau^2g^{\overline{\beta\gamma}}(\del_\mu \tau^{-2}g_{\overline{\alpha\gamma}})(\del_\beta \tau^{-1}\tilde{A}_{\overline{\nu}}) \\
    &= \frac12g^{\overline{\beta\gamma}}g_{\overline{\alpha\gamma}}(\tau \tau\del_\beta \tau^{-1}\tilde{A}_{\overline{\nu}})(\tau^2 \tau \del_\mu \tau^{-2}g_{\overline{\alpha\gamma}}).
\end{align*}
The differential operators appearing here are of the form $\tau^n \tau\del_\nu \tau^{-n} = \tau\del_\nu + \tau^{n}[\tau\del_\nu,\tau^{-n}]$. As the commutator is $-n$ if $\nu = \tau$ and $0$ otherwise, the term will hence be smooth, because $g,A$ are smooth as $0$-sections. Any other of the lower order terms may be treated in the same manner. Similarly
\begin{equation}\label{eq:expansionUpsilonM}
\begin{split}
    -\rm{d}\Upsilon^M(g;A)_\nu &= -\del_\nu \Tr_g \delta_{g_{(\infty)}}^*\tilde{A}
    \\&= -\del_\nu (g^{\mu\alpha} \ {}^{g_{(\infty)}}\nabla_\alpha\tilde{A}_\mu)
    \\&= -\del_\nu (g^{\mu\alpha}(\del_\alpha \tilde{A}_\mu - A_\beta \Gamma(g_{(\infty)})\indices{^\beta_{\alpha\mu}})
    \\&= - g^{\mu\alpha}\del_\nu \del_\alpha \tilde{A}_\mu + \text{l.o.t.}
\end{split}
\end{equation}
All lower order terms may be analyzed as before. Adding \eqref{eq:expansiondelta_gdA_nu} and \eqref{eq:expansionUpsilonM} one indeed sees that $P_L$ is in the required form
\[
P_L(g;A) = \Box_g \Tilde{A} + P_1(\tau,x,\Tilde{g},\Tilde{A})
\]

We turn towards $P_{DT}$. We may once more put the $2(-\Lambda g - 2T(g,A))$ terms into the $P_1$ term, as they involve no second derivatives of $\tilde{g}$ or $\tilde{A}$ and are, as argued before for the terms in $P_L$, smooth. The other two terms $2(\Ric(g) - \delta_g^* \Upsilon(g))$ involve second derivatives of $\tilde{g}$, and we need to show that these second derivatives are exactly given by $\Box_g$, with the first order terms fitting into $P_1$.
One may calculate that (where we again refer to \cite[\S 4.1]{asy-ds} for details)
\begin{alignat*}{2}
    &2 \delta_g^* \Upsilon(g)_{\overline{\mu\nu}} &&= \tau^2(g)^{\gamma\rho}\left(\del_\mu\del_\gamma {\tilde{g}}_{\nu\rho} + \del_\nu\del_\gamma {\tilde{g}}_{\mu\rho} - \del_\mu\del_\nu {\tilde{g}}_{\gamma\rho}\right) + \text{l.o.t.} \\
    &2\Ric(g)_{\overline{\mu\nu}} &&= \tau^2(g)^{\gamma\rho}(\del_\mu\del_\gamma {\tilde{g}}_{\nu\rho}  + \del_\nu\del_\gamma {\tilde{g}}_{\mu\rho} - \del_\mu\del_\nu {\tilde{g}}_{\rho\gamma}  - \del_\gamma\del_\rho {\tilde{g}}_{\mu\nu}) + \text{l.o.t.}
\end{alignat*}
Combining this, we see that the second derivatives in $\tilde{g}$ in $P_{DT}$ indeed cancel, except for the ones in $\Box_g \Tilde{g}$. \\
In conclusion, we may apply \cref{prop:UniquenessAndExistenceQuasilinearWave} to the system of quasilinear wave equations
\[
\begin{pmatrix}
    P_{DT} \\
    P_L
\end{pmatrix}(g;A) = 0.
\]
This proves \cref{prop:solutionToGaugeFixedEq}.

It remains to show that the solutions $g,A$ satisfy the Einstein--Maxwell system. Apply $\delta_g$ to the equation $P_L(g;A) = 0$. Because $\delta_g^2=0$, the quantity $\Upsilon^M(g;A)$ itself solves
\begin{equation}
    \delta_g \rm{d}\Upsilon^M(g;A) = 0.
\end{equation}
It is easy to see by the previous methods that this is a quasilinear, in fact, linear wave equation in the form of \cref{prop:UniquenessAndExistenceQuasilinearWave} for $\Upsilon^M(g;A)$.
Therefore, as long as $\Upsilon(g;A) \in \tau^\infty\mc{C}^\infty$, the uniqueness part yields $\Upsilon^M(g;A) = 0$. But this is the case, since $A - A_{(\infty)} \in \tau^\lambda \mc{C}^\infty$ for all $\lambda$ implies
\begin{equation}\label{eq:UpsilonMCinfty}
    \begin{split}
    \Upsilon^M(g;A) &\equiv \Upsilon^M(g; A_{(\infty)}) \mod \tau^{2\lambda}\mc{C}^\infty
    \\ &= 0.
    \end{split}
\end{equation}
In conclusion Maxwell's equations are satisfied. By the same calculations as in \cref{lemma:OrderOfdeltaT}, we see that this also implies that $\delta_gG_g T(g,A) = 0$. Applying $\delta_gG_g$ to $P_{DT}$ hence yields, by the second Bianchi identity,
\begin{equation}
    \delta_g G_g \delta_g^* \Upsilon(g) = 0.
\end{equation}
This can again be seen to be a quasilinear wave equation for $\Upsilon(g)$. But as $\Upsilon(g) \in \tau^\infty\mc{C}^\infty$ by a calculation analogous to \eqref{eq:UpsilonMCinfty} another application of \cref{prop:UniquenessAndExistenceQuasilinearWave} shows $\Upsilon(g) = 0$. Thus Einstein's field equations are satisfied too. This concludes the construction.

%% file: finalremarks.tex
\section{Additional remarks}
    In the source material, the domain of existence for small masses and charges, the necessity of the balance conditions \cite[\S 3.4]{blackholegluing}, and gluing with non-compact spatial topology was discussed \cite[\S 3.5]{blackholegluing}. These same considerations apply here without many changes, as we shall now briefly discuss. \\
    \begin{remark}[Domain of existence for small masses and charges]\label{remark:DomainOfExistence}
        One may consider the maximal globally hyperbolic development of the constructed solutions $g,A$. It was shown in \cite[\S 3.3]{blackholegluing} that for glued Schwarzschild--de Sitter black holes with very small masses, Cauchy stability implies that the cosmological horizons of at least two black holes must intersect non-trivially. The same holds here. Namely, we may consider, for $p_i,\mf{m}_i,Q_i$ satisfying the mass and charge balance conditions, the charged black hole metrics with parameters $p_i, \lambda \mf{m}_i, \lambda Q_i$, where $\lambda > 0$. Similar arguments as for Schwarzschild--de Sitter black holes show that the formal solutions $g_{(\infty)},A_{(\infty)}$ may be chosen such that $g_{(\infty)} - g_{(3),\lambda} \in \lambda \tau^3\mc{C}^\infty$ and $A_{(\infty)} - A_{(2),\lambda} \in \lambda \tau^2\mc{C}^\infty$, where $g_{(3),\lambda}, A_{(2),\lambda}$ are the naively glued metric and potential as in \eqref{eq:DefinitionMetricg_(3)}, \eqref{eq:DefinitionPotentialA_(2)}. \\
        As these formal solutions will converge to the de Sitter metric, respectively, the (zero) de Sitter potential when $\lambda \rightarrow 0$, the same geometrical considerations as in \cite{blackholegluing} work to show that the cosmological horizons must intersect.
    \end{remark}

    To discuss the necessity of the balance condition, we recall, for $\alpha \in \R$, the functions \emph{conormal} relative to $\tau^\alpha L^\infty(M)$. This is the space
    \[
    \mc{A}^\alpha \coloneqq \{u \in \mc{C}^\infty(M^\circ): P(\tau^{-\alpha} u) \in L^\infty(M) \ \forall \ P \in \Diff_b(M)\},
    \]
    where $\Diff_b(M)$ is the space of $b$-differential operators on $M$. These are finite linear combinations of compositions of $b$ vector fields, which, in turn, are sections of the vector bundle with local frames given by $\tau\del_\tau, \del_{x^i}$. \\
    The space $\mc{A}^\alpha$ is a $\mc{C}^\infty(M)$ module, and we have for $f \in \mc{A}^\alpha, g \in \mc{A}^\beta$ that $fg \in \mc{A}^{\alpha +\beta}$. Therefore, \cref{lemma:StabilityOfLinearization} and \cref{lemma:DifferentialOperatorsUpToLinearization} hold with all instances of $\tau^k\mc{C}^\infty$ replaced by $\mc{A}^k$. \\
    Note also that $\tau^\alpha (\log \tau)^l \in \mc{A}^{\alpha-\epsilon}$ for all $\epsilon >0$ and $l>0$.
    
    \begin{theorem}[Necessity of the balance conditions, {\cite[Theorem 3.4]{blackholegluing}}]\label{theorem:NecessityOfBalanceConditions}
        Let $(p_i , \mf{m}_i, Q_i) \in \sph^3 \times \R \times \R$ (with pairwise distinct $p_i$) and assume $g,A$ satisfy (1) - (4) of \cref{theorem:BlackHoleGluing}. If, for some $\epsilon > 0$, we have 
        \begin{align*}
        g-g_{dS} &\in \tau^3(\log \tau)\mc{C}^\infty + \tau^3\mc{C}^\infty + \mc{A}^{3+\epsilon}(M; S^2 \nT^*M) \\
        A &\in \tau^2(\log \tau)\mc{C}^\infty + \tau^2\mc{C}^\infty + \mc{A}^{2+\epsilon}(M; \nT^*M)
        \end{align*}
        then the $(p_i , \mf{m}_i, Q_i)$ satisfy the charge and mass balance conditions.
    \end{theorem}

    \begin{proof}
        Let $g_{(3)}$ and $A_{(2)}$ be the naively glued metrics as in \eqref{eq:DefinitionMetricg_(3)} and \eqref{eq:DefinitionPotentialA_(2)}. Then the assumption gives us 
        \begin{align*}
            g - g_{(3)} &= \tau^3(\log \tau)g_l + \tau^3 g_3 + \Tilde{g} \\
            A - A_{(2)} &= \tau^2(\log \tau)A_l + \tau^2 A_2 + \Tilde{A}
        \end{align*}
        for some $g_l, g_3 \in \mc{C}^\infty(\del M; S^2 \nT^*M)$, $\Tilde{g}\in \mc{A}^{3 + \epsilon}(M; S^2 \nT^*M)$,
        $A_l, A_2 \in \mc{C}^\infty(\del M; \nT^*M)$ and $\Tilde{A}\in \mc{A}^{2 + \epsilon}(M; \nT^*M)$. \\
        We first show the necessity of the charge balance condition.
        By \cref{lemma:StabilityOfLinearization}, \cref{lemma:DifferentialOperatorsUpToLinearization} and because $\deld(g,A) = 0$, we have
            \begin{multline}\label{eq:ErrorDeltaDNecessity}
            K_{0,g_{dS},0}(0,\tau^2(\log \tau)A_l + \tau^2 A_2 + \Tilde{A}) + \tau^3\rm{Err}_{\deld} \\ = K_{0,g_{dS},0}(g-g_{(3)}, A-A_{(2)}) + \tau^3\rm{Err}_{\deld} \in \mc{A}^{4- \delta}
            \end{multline}
        for all $\delta > 0$. Here we have again used that $K_{0,g_{dS},0} $ vanishes when evaluated on symmetric $2$-tensors. \\
        Now, because $I(K_{0,g_{dS},0}, \lambda)_N = 0$, we may write $(K_{0,g_{dS},0})_N = \tau K'$, where $K' \in \Diff_b(M)$. This shows that $(K_{0,g_{dS},0} \Tilde{A})_N \in \mc{A}^{3+\epsilon}$.
        The order $\tau^3$ part of the normal component of \eqref{eq:ErrorDeltaDNecessity} is therefore
        \begin{align*}
            0 &= (\rm{Err}_{\deld})_N + (\del_\lambda I(K_{0,g_{dS},0}[\tau], \lambda)\vert_{\lambda = 2}A_l)_N + (I(K_{0,g_{dS},0}[\tau], 2)A_{(2)})_N \\
            &=(\rm{Err}_{\deld})_N -\delta_h (A_l)_N -\delta_h (A_{(2)})_N
        \end{align*}
        The two last terms of this vanish when integrating, so we are once again led to the charge balance condition. \\
        As in the proof of \cref{prop:CorrectionOfObstructionError}, we write $P = P_0 - 4T$, so we treat the $T$ terms separately. The reason for this is that because $A-A_{(2)} \in \tau^2\mc{C}^\infty$, $P$ only agrees with its linearization up to order $\tau^4$. We get
        \begin{align*}
            0 \equiv L_{0,g_{dS}}(\tau^3(\log \tau)g_l + \tau^3(\log \tau)g_3 + \Tilde{g}) + \tau^4 \rm{Err}_{P_0} - 4T(g,A) \mod \mc{A}^{6-\delta} \\
            \equiv L_{0,g_{dS}}(\tau^3(\log \tau)g_l + \tau^3(\log \tau)g_3 + \Tilde{g}) + \tau^4 \rm{Err}_{P_0} - 4T(g_{dS},A) \mod \mc{A}^{5}
        \end{align*}
        for all $\delta > 0$. Now, the $-4T(g,A)$ term is made up of terms in $\tau^4(\log \tau) \mc{C}^\infty$, $\tau^4(\log^2 \tau) \mc{C}^\infty$, $\mc{A}^{4+\epsilon}$ et cetera, coming from the fact that $T$ is quadratic in $A$. As in the part after \cref{lemma:CalculationErrP0s}, the normal-tangential part of the order $\tau^4$ component of this vanishes, so the $T$ terms do not contribute there. For the rest of the proof, one may therefore follow \cite[\S 3.5]{blackholegluing} ad verbatim.
    \end{proof}

    \begin{remark}[Gluing with noncompact spatial topology]\label{remark:NoncompactTopology}
        The two balance conditions are a topological artefact of the conformal boundary of de Sitter space. That is, the requirement for the vanishing of the integral of the error to $\deld$ in \eqref{eq:chargeconditionOne} and the orthogonality to the cokernel of the error to $P$ in \eqref{eq:masscondition} stem from the conformal boundary $\sph^3$ being compact. 
        If we instead use the upper half-space coordinates $M_u \coloneqq [0,\infty)_{\tau'} \times \R^3_x$ from \eqref{deSitter:UpperHalfSpace}, the conformal boundary is $\R^3$ with metric $\frac{3}{\Lambda} dx^2$. The point $-p_0 = (-1,0,0,0)$, or some point $p_\infty$ after pulling back along a rotation $R \in \mc{SO}(4)$, is not covered. Then a correction to the potential $A$ for $\deld$ \emph{always} exists by the vanishing of the $n$-th cohomology group of noncompact connected oriented manifolds (\cite[Theorem 17.32]{LeeManifolds}). Indeed, by this vanishing, we may choose an $(n-1)$-form $\omega$ to cancel the error as before. While this $\omega$ does not necessarily vanish near the $\overline{V_{p_i}}$, we can make sure it does by adding further exterior derivatives of $(n-2)$-forms (supported on a small neighborhood of $\overline{V_{p_i}}$) to cancel it there. The existence of these further corrections in turn follows from the Poincaré lemma with compact support \cite[Lemma 17.27]{LeeManifolds}. Of course, the resulting potential correction may be singular at infinity, i.e. at the point $p_\infty$. 
        Similarly, as argued in \cite[\S 3.5]{blackholegluing}, the correction to the metric will always exist as well if one allows such behavior at infinity. In summary:
    \end{remark}
    \begin{theorem}\label{theorem:BlackHoleGluingNonCompact}
        Let $N \in \N$ and let $p_i \in \R^3, \mf{m}_i,Q_i \in \R$ for $1 \leq i \leq N$. Let $V_{p_i} \subset \sph^3 = \del M$ be a neighborhood of $p_i$ with $p_i$ removed and assume that $\overline{V_{p_i}} \cap \overline{V_{p_j}} = 0$ for all $i \neq j$. Then there exists a neighborhood $U$ of $\del M_u \setminus \{p_1, \dots, p_N\}$, a Lorentzian $0$-metric $g \in \mc{C}^\infty(U;S^2 \nT_U^*M_u)$ and a vector potential $A \in \mc{C}^\infty(U;\nT^*M_u)$ with the following properties:
        \begin{enumerate}
            \item $g$ and $A$ satisfy the Einstein--Maxwell equations $\Ric(g) - \Lambda g = 2T$, $\delta_g \rm{d} A = 0$,
            \item near $V_{p_i}$ we have $g = g_{p_i,\mf{m_i},Q_i}$, $A = A_{p_i,Q_i}$,
            \item $g-g_{dS} \in \tau^3\mc{C}^\infty(U; S^2 \nT_U^*M_u)$,
            \item $A \in \tau^2 \mc{C}^\infty(U; \nT_U^*M_u)$.
        \end{enumerate}
    \end{theorem}

%% file: KerrNewmanDeSitter.tex
\section{Gluing rotating charged black holes}\label{sec:RotatingChargedBlackHoles}
Rotating charged black holes are modeled by the Kerr-Newman-de Sitter (KNdS) metric and electromagnetic potential. In Boyer-Lindquist coordinates (see for example \cite[\S 3.2]{HintzKerrNewmanDeSitter}), the metric for a given mass $\mf{m}$, charge $Q$, and angular momentum $\mathbf{a}$ is given by
\begin{equation}\label{eq:KNdSmetricBL}
    g_{\mf{m},Q,\mathbf{a}} = - \frac{\Delta_r}{\rho^2} \Big( \rm{d}{t_0} - \frac{\mathbf{a} \sin^2 \theta_0}{\Delta_0} \rm{d}{\phi_0}\Big)^2 + \frac{\rho^2}{\Delta_r} \dl[2]{r_0} + \frac{\rho^2}{\Delta_\theta} \dl[2]{\theta_0} + \sin^2 \theta_0 \frac{\Delta_\theta}{\rho^2}\Big(\mathbf{a} \rm{d}{t_0} - \frac{r_0^2 + \mathbf{a}^2}{\Delta_0}\rm{d}{\phi_0}\Big)^2,
\end{equation}
where, for $\lambda = \sqrt{\Lambda/3}a$,
\begin{align*}
    &\Delta_0 = 1 + \lambda^2,& &\rho^2 = r_0^2  + \mathbf{a}^2\cos^2 \theta_0, \\
    &\Delta_r = (r_0^2 + \mathbf{a}^2)\Big( 1- \frac{\Lambda r_0^2}{3}\Big) - 2 \mf{m}r_0 + \Delta_0 Q^2,& &\Delta_\theta = 1 + \lambda^2 \cos^2 \theta_0.
\end{align*}
The electromagnetic potential is
\begin{gather}\label{KNdSpotentialBL}
    A_{Q, \mathbf{a}} = -\frac{Q r_0}{\rho^2}(\rm{d}{t_0} - \mathbf{a} \sin^2 \theta_0 \rm{d}{\phi_0}),
\end{gather}
As in \cite[Appendix B]{Schlue2015Mar} (or again \cite[\S 4.1]{blackholegluing}) one may take comoving coordinates
\begin{align*}
    &t = t_0,& &\phi = \phi_0 - \frac{\Lambda}{3} a t_0, \\
    &r^2 = \frac{1}{\Delta_0}(r_0^2 \Delta_\theta + \mathbf{a}^2\sin^2 \theta_0),& &r \cos\theta = r_0 \cos \theta_0,
\end{align*}
to view the Kerr-Newman-de Sitter spacetime as a perturbation of de Sitter space. Indeed, transforming the de Sitter metric in the upper half space coordinates \ref{deSitter:UpperHalfSpaceMetric} according to these rules it takes the form of \ref{eq:KNdSmetricBL} with $\mf{m} = Q = 0$, so
\begin{align*}
    g_{\mf{m},Q,\mathbf{a}} =  g_{dS} + c_{\mf{m}, Q, \mathbf{a}},
\end{align*}
where
\begin{equation}\label{eq:cmqa}
    c_{\mf{m}, Q, \mathbf{a}} = \frac{2\mf{m}r_0 - \Delta_0 Q^2}{\rho^2}\Big(\rm{d}{t_0} - \frac{\mathbf{a}\sin^2 \theta_0}{\Delta_0}\rm{d}{\phi_0}\Big)^2 + \frac{(2\mf{m}r_0 - \Delta_0 Q^2)\rho^2}{\Delta_r\vert_{\mf{m}= 0, Q=0} \Delta_r}\dl[2]{r_0}
\end{equation}

Transforming the KNdS metric and potential back to the upper half space coordinates $(\tau,x)$ \eqref{deSitter:UpperHalfSpace} in turn, one sees as before that the metric is defined on a neighborhood of $\sph^3 \setminus \{p_0,-p_0\} \subset \R^4$, where $p_0 = (1,0,0,0)$ is the north pole. The rotation is a rotation around the axis $\mathbf{a}_0 = (0,0,0,1)$ in $p_0^\bot$, i.e. along the 3-spheres' Killing vector field
\[
\mathfrak{a}_0 = \begin{pmatrix}
    0 & 0 & 0 & 0 \\
    0 & 0 & \mathbf{a} & 0 \\
    0 & -\mathbf{a} & 0  & 0 \\
    0 & 0 & 0 & 0
\end{pmatrix} \in \mathfrak{so}_4.
\]
As discussed in \cite[\S 4.1]{blackholegluing}, one obtains KNdS black holes on other points $p_i$, rotating in $p_i^\bot$ around other axes $\mathbf{a}_i$ (i.e. vector fields $\mathfrak{a}_i \in \mathfrak{so}_4$ with $\mathfrak{a}_ip_i = 0$) by pulling back along a rotation $R$ with $Rp_i = p_0$ and $R\mathbf{a}_i = \mathbf{a}_0$. We call so obtained metrics $g_{p_i, \mathfrak{m}_i, Q_i, \mathfrak{a}_i}$ and potentials $A_{p_i,, Q_i, \mathfrak{a}_i}$. Recall also the inner product defined on $\mathfrak{so}_4$ via
\[
\left<\mathfrak{a}_1,\mathfrak{a}_2 \right> = \sum_{i < j} (\mathfrak{a}_1)_{ij}(\mathfrak{a}_2)_{ij}.
\]
We show

\begin{theorem}\label{theorem:blackholegluingKNdS}
Let $N \in \N$ and let $p_i \in \sph^3, \mf{m}_i,Q_i \in \R, \mathfrak{a}_i \in \mathfrak{so}_4$ for $1 \leq i \leq N$ be such that $\mathfrak{a}_i p_i = 0$. Let $\lambda_i = \sqrt{\Lambda/3}|\mf{a}_i|$ and define the \emph{effective charge and mass} as
\begin{equation}
\begin{split}
    Q_{\mathrm{eff},i} &= Q_i(1 - \lambda_i \arctan \lambda_i), \\
    m_{\mathrm{eff},i} &= \frac{m_i}{1 + \lambda_i^2}.
\end{split}
\end{equation}
Assume that they satisfy the charge balance condition
\begin{equation}
    \sum_{i=1}^N Q_{\mathrm{eff},i} = 0
\end{equation}
and the mass and rotation balance conditions
\begin{align}
    &\sum_{i=1}^N m_{\mathrm{eff},i}p_i = 0 \label{eq:massbalanceKNdS}\\
    &\sum_{i=1}^N  m_{\mathrm{eff},i}\mathfrak{a}_i = 0 \label{eq:rotationbalanceKNdS}
\end{align}
Let $V_{p_i} \subset \sph^3 = \del M$ be a neighborhood of $p_i$ with $p_i$ removed and assume that $\overline{V_{p_i}} \cap \overline{V_{p_j}} = 0$ for all $i \neq j$. Then there exists a neighborhood $U$ of $\del M \setminus \{p_1, \dots, p_N\}$, a Lorentzian $0$-metric $g \in \mc{C}^\infty(U;S^2 \nT_U^*M)$, and a vector potential $A \in \mc{C}^\infty(U;\nT^*M)$ with the following properties:
        \begin{enumerate}
            \item $g$ and $A$ satisfy the Einstein-Maxwell equations $\Ric(g) - \Lambda g = 2T$, $\delta_g \rm{d}A = 0$,
            \item near $V_{p_i}$ we have $g = g_{p_i,\mf{m_i},Q_i, \mathfrak{a}_i}$, $A = A_{p_i,Q_i, \mathfrak{a}_i}$,
            \item $g-g_{dS} \in \tau^3\mc{C}^\infty(U; S^2 \nT_U^*M)$,
            \item $A \in \tau^2 \mc{C}^\infty(U; \nT_U^*M)$
        \end{enumerate}
\end{theorem}

\begin{remark}
    The scaling factor $1-\lambda \arctan \lambda$ in the effective charge is somewhat surprising in that it differs from the factor of effective mass, and may be zero or even negative. It is positive for subextremal Kerr-Newman-de Sitter at least, as there $0 \leq \lambda \leq 2 - \sqrt{3}$; see, for example, \cite[Figure 2]{davey_strong_2024}.
\end{remark}
To prove this theorem we, as in the RNdS case, start by first naively gluing KNdS black holes into the conformal future and calculate the resulting errors to $P$ and $\deld$ defined in \eqref{eq:definitionPanddeltad}.

For this we again set $\tau_s = 1/r$, so that $g_{\mf{m},Q,\mathbf{a}}$ and $A_{Q,\mathbf{a}}$ become $0$-tensors smooth up to $\tau_s = 0$. We transform the KNdS potential back to the coordinates $(t,r,\theta,\phi)$. This yields (using $\rm{d}t_0 = \rm{d}t, \rm{d}\phi_0 = \rm{d}\phi +(\Lambda/3) \mathbf{a}\rm{d}t$)
\[
\begin{split}
A_{Q, \mathbf{a}} &= - \frac{Q r_0}{\rho^2}\Big(\Big(1-\frac{\Lambda}{3}\mathbf{a}^2\sin^2\theta_0\Big) \rm{d}t - \mathbf{a}\sin^2 \theta_0 \rm{d}\phi\Big) \\
&\equiv - Q\sqrt{\frac{\Delta_\theta}{\Delta_0}}\tau_s^2\Big(\Big(1-\frac{\Lambda}{3}\mathbf{a}^2\sin^2\theta_0\Big) \frac{\rm{d}t}{\tau_s} - \mathbf{a}\sin^2 \theta_0 \frac{\rm{d}\phi}{\tau_s}\Big) \mod \tau_s^3\mc{C}^\infty
\end{split}
\]
Here we have used that $r/r_0 = \sqrt{\Delta_\theta/\Delta_0} + \mc{O}(\tau_s^2)$ and $r_0/\rho^2 = 1/r_0 + \mc{O}(\tau_s^3)$.
This vanishes at the conformal boundary to the order $\tau_s^2$, as in the RNdS case. \\
For lowest order terms of $c_{\mf{m},Q,\mathbf{a}}$ on the other hand, observe that the additional terms involving $Q$ that we get compared to the Kerr-de Sitter case do \emph{not} change the lowest order terms in its Taylor expansion as a $0$-tensor. This is because the $\Delta_0Q^2$ term in $2\mf{m}r_0 - \Delta_0 Q^2$ is one order higher in $\tau_s$ than $2\mf{m}r_0$; compare with the Taylor expansion in the Reissner-Nordström case in \eqref{eq:gamma3andgamma4}. These lowest order terms were calculated in \cite[Lemma 4.1]{blackholegluing}; $c_{\mf{m},Q,\mathbf{a}}$ vanishes to order $\tau_s^3$ at the conformal boundary, as in the RNdS case. \\
Recall from \cref{prop:CorrectionOfObstructionError} and its proof that only the lowest order terms of the difference compared to de Sitter of the potential and metric mattered for finding the obstructions and resulting balance conditions. Also, the resulting balance conditions on the charge and mass were independent of each other, since the lowest order term of the metric perturbations' Taylor series did not involve any $Q$ terms and vice versa for the potential. As the same is the case here, there are therefore no additional conceptual difficulties in the KNdS case.

We get the following error to $\deld$ if we try to glue a single KNdS black hole into de Sitter space:
\begin{lemma}
    Let $\chi \in \mc{C}^\infty(\R_t)$ be a smooth cutoff function that is $1$ for large $t$ and let 
    \[
    \begin{split}
    g_{(3)} &= \chi(t)g_{\mf{m},Q,\mathbf{a}} + (1-\chi(t))g_{dS} \\
    A_{(2)} &= \chi(t)A_{Q,\mathbf{a}} 
    \end{split}
    \]
    be the naively glued metric and potential. Then
    \begin{equation}
        \deld(g_{(3)}, A_{(2)}) \equiv \tau_s^3 \mathrm{Err}_{\deld,s} \mod \tau_s^4 \mc{C}^\infty,
    \end{equation}
    where $\mathrm{Err}_{\deld, s} = - \frac{3Q}{\Lambda} \sqrt{\frac{\Delta_\theta}{\Delta_0}} \Big(1 - \lambda^2 \sin^2\theta\Big) \chi'(t) \frac{\rm{d}\tau_s}{\tau_s}$.
\end{lemma}
To show this, one may follow the proof of \cref{lemma:ErrorDeltaD}. \\
As in \cref{prop:correctionDeltaD}, the obstruction to solving away this error is given by the integral of the normal component of the error over the boundary $(\del M_s, h_s) = (\R_t \times \sph^2, \Lambda^2/9 \rm{d}t^2 + \Lambda/3 g_{\sph^2}$). As calculated in \cite[\S 4.1]{blackholegluing} we have
\[
dh_s = \frac{\Lambda^2}{9} \sqrt{\frac{\Delta_0}{\Delta_\theta}} \Delta_\theta^{-1} \rm{d}t_0 \sin \theta_0 \rm{d}\theta_0 \rm{d}\phi_0,
\]
so the obstruction is
\begin{equation}
\begin{split}
    \int_{\del M_s} (\mathrm{Err}_{\deld,s})_N dh_s 
    &= - \frac{3Q}{\Lambda} \int_{0}^{2\pi} \int_0^\pi \int_{-\infty}^\infty \frac{1 - \lambda^2 \sin^2 \theta_0}{1 + \lambda^2 \cos^2 \theta_0}\chi'(t)  \rm{d}t_0 \sin \theta_0 \rm{d}\theta_0 \rm{d}\phi_0 \\
    &= -\frac{12 \pi Q}{\Lambda} (1 - \lambda \arctan \lambda) \\
    &= -\frac{12 \pi}{\Lambda} Q_{\mathrm{eff}}.
\end{split}
\end{equation}
We hence get the following analogue of \cref{lemma:chargebalancecondition}
\begin{lemma}
    Suppose we have tuples $(p_i, \mathfrak{m_i}, Q_i, \mathfrak{a}_i) \in \sph^3 \times \R \times \R \times \mathfrak{so}_4$ with $\mathfrak{a}_i p_i = 0$. Let $\chi_i$ be (arbitrary) cutoff functions on $\sph^3$ which are identically $1$ near $p_i$ and $0$ near $-p_i$. Set $\rm{Err}_{\deld, p_i, Q_i, \mathfrak{a}_i} = \tau^{-3} \deld(\chi_i g_{p_i,\mf{m}_i, Q_i,\mathfrak{a}_i} + (1-\chi_i) g_{dS}, \chi_i A_{Q_i,\mathfrak{a}_i})(\tau\del \tau)\vert_{\tau = 0}$ and $\rm{Err}_{\deld} = \sum_{i=1}^N \rm{Err}_{\deld,p_i, Q_i,\mathfrak{a}_i}$. Then $\rm{Err}_{\deld}$ is a normal $1$-form and 
    \[
    \int_{\sph^3} (\rm{Err}_{\deld})_N \rm{d}g_{\sph^3} = 0
    \]
    if and only if the charge balance condition
    $\sum_{i=1}^N Q_{\rm{eff},i} = 0$
    is satisfied.
\end{lemma}
Now, by the same considerations as in the RNdS case, the error to $P$ will be the same error as for Kerr-de Sitter black holes. This leads to the same balance for the masses and rotations conditions as for Kerr-de Sitter black holes. These were shown to be the ones in \eqref{eq:massbalanceKNdS} and \eqref{eq:rotationbalanceKNdS} in \cite[\S 4.1]{blackholegluing}. The rest of the gluing construction unfolds without issues, which finishes the proof of \cref{theorem:blackholegluingKNdS}. 